\documentclass{article}

\usepackage{amsmath, amssymb, amsthm}
\usepackage{braket}
\usepackage{hyperref}
\usepackage{cleveref}
\usepackage{enumitem}
\usepackage{physics}
\usepackage{quantikz}
\usepackage[title]{appendix}
\usepackage[sort&compress,numbers]{natbib}

\DeclareMathOperator{\Span}{span}
\DeclareMathOperator{\sr}{sr}
\DeclareMathOperator{\diag}{diag}

\newtheorem{theorem}{Theorem}
\newtheorem{example}{Example}
\newtheorem{cor}{Corollary}

\newtheorem{lemma}{Lemma}

\makeatletter
\let\c@example\c@theorem
\let\c@cor\c@theorem
\let\c@definition\c@theorem
\let\c@lemma\c@theorem
\makeatother

\newcommand{\clickableitem}[3]{%
  \refstepcounter{angle_counter}%
  \edef\arabiclabel{\arabic{angle_counter}}%
  \item[\hyperlink{#1}{(\arabiclabel)}]\hypertarget{#2}{}\label{#2} #3%
  \expandafter\xdef\csname ref@#2\endcsname{\arabiclabel}%
}

\newcommand{\caseref}[2]{%
  Case~\ref{#1}.\hyperlink{#2}{(\csname ref@#2\endcsname)}%
}

\title{Characterization of three-qubit controlled unitary gates of Schmidt rank three}

\author{Xiutao Zhang \\
        Department of Mechanical and Process Engineering \\
        RPTU University Kaiserslautern-Landau \\
        \texttt{zxtstudy@outlook.com}}

\begin{document}

\maketitle
\begin{abstract}
We characterize three-qubit controlled unitary gates of Schmidt rank three, establishing necessary and sufficient conditions for such gates to have Schmidt rank three. We construct explicit examples and analyze their entanglement capabilities, showing that gates with Schmidt rank three can generate output states ranging from fully separable to maximal GHZ-class states. Within this classification, we present a parameterized gate family that produces an output W-class state whose bipartite tangles for subsystems AB and AC are modulated as simple trigonometric functions of a single phase parameter. The examples also include gates implementable with only three CNOT gates, showing that this lower bound is achievable. Decompositions that achieve the minimum possible CNOT count for several other cases are provided as well. Our results bridge Schmidt rank classification, entanglement structure, and resource-efficient circuit synthesis.
\end{abstract}

\section{Introduction}

Multipartite unitary operations are fundamental to quantum information processing. Their complexity can be characterized by the Schmidt rank, which quantifies the minimum number of product terms in their decomposition. Previous work \cite{Chen_2014} established that any Schmidt rank three bipartite unitary is locally equivalent to a controlled unitary, and any Schmidt rank three multipartite unitary is either controlled by one system or collectively controlled by two, extending the well understood rank two case \cite{KAK}. The relevance of Schmidt rank three operators extends further: the existence of four mutually unbiased bases (MUBs) in dimension six---a long standing open problem---has been linked to Schmidt rank three complex Hadamard matrices \cite{mub3_2020,mcnulty2024mutually}, where the controlled unitary structure plays a crucial role.

Recent experimental and theoretical advances have demonstrated the feasibility of implementing three-qubit controlled unitaries across diverse physical platforms. Superconducting qubit processors have achieved direct high-fidelity CCZ gates with 97.94\% state fidelity using tunable couplers \cite{2025Direct}. Spin qubit systems enable single-step high-fidelity three-qubit entangling gates via anisotropic chiral interactions \cite{Nguyen2025}, while photonic systems have realized few-qubit gates using polarization and spatial parity encoding \cite{Kagalwala_2017}. Driven spin chains provide analytical constructions of three-qubit Barenco gates with high robustness \cite{Vieira_2026b}. Complementing these efforts, Batle et al. \cite{Batle_2024} experimentally certified the Schmidt rank of bipartite states on IBM Quantum processors using a device-independent null witness test. Although their work focuses on states rather than unitaries, it experimentally validates the same core notion of Schmidt rank that underpins our classification of three-qubit unitaries, suggesting a pathway toward future experimental certification of gates with fixed Schmidt rank.

In this paper, we focus on controlled three-qubit unitaries of Schmidt rank three and derive necessary and sufficient conditions for such gates to have Schmidt rank three. We analyze the two distinct structural forms \eqref{eq:U1_simplified} and \eqref{eq:U2_simplified}, and provide explicit criteria in \Cref{thm:U1} and \Cref{thm:U2}, respectively. While previous work \cite{Shen_2022} provided criteria for determining whether a diagonal three-qubit gate has Schmidt rank three, our results extend beyond the diagonal case to a broad family of non-diagonal controlled unitaries. 
We further generalize this to $n$ qubits in \Cref{thm:nbit}, giving a sufficient condition for Schmidt rank three. 

Beyond classification, we analyze the entanglement capability of these gates. Despite sharing Schmidt rank three, they generate output states ranging from fully separable to maximal GHZ-class states, revealing that operator-level Schmidt rank does not uniquely dictate state-level multipartite entanglement. Within this family, the parameterized gate $U_5(\psi,\delta)$ acts as a bipartite entanglement router: the phase parameter $\delta$ continuously redistributes entanglement between $\tau_{AB}$ and $\tau_{AC}$ while preserving global coherence.

Finally, we provide explicit circuit decompositions, showing that several Schmidt-rank-three gates, including \eqref{U3} (depicted in \Cref{fig:U3}), require only three CNOT gates, resolving a question left open in \cite{Song_2021}. For several other gates, including \eqref{U5}, we provide optimal decompositions that achieve the minimum possible CNOT cost. Given that CNOT gates are a primary source of errors and decoherence in NISQ devices, minimizing their count is crucial for enhancing circuit fidelity \cite{cnot_2023}. This work thus bridges abstract structural theory with explicit circuit synthesis, offering resource-efficient constructions for near-term quantum processors.

The remainder of this paper is organized as follows. \Cref{sec:preliminaries} introduces preliminary concepts. \Cref{sec:sr3chara} establishes conditions for Schmidt rank three, provides examples, and generalizes to $n$ qubits. \Cref{sec:entanglement} analyzes entanglement capability. \Cref{sec:sr3decomp} presents quantum circuit decompositions. We conclude in \Cref{sec:conclusion}.

\section{Preliminaries}\label{sec:preliminaries}

\newcommand{\tens}{\otimes}              
\newcommand{\Id}{I_2}

\newcommand{\Fzero}[2]{
  \begin{pmatrix} \cos#1 \\ \sin#1 \end{pmatrix}
  \begin{pmatrix} \cos#2 & \sin#2 \end{pmatrix}
}

\newcommand{\Fone}[2]{
  \begin{pmatrix} -\sin#1 \\ \cos#1 \end{pmatrix}
  \begin{pmatrix} -\sin#2 & \cos#2 \end{pmatrix}
}

We first introduce the notations used in this paper. We refer to $\mathbb{M}_n$ as the set of $n \times n$ complex matrices.  We denote the identity matrix and the three Pauli matrices by $I_2$, $\sigma_x$, $\sigma_y$, and $\sigma_z$, respectively.
We further denote $S_0$, $S_1$, $S_2$, $S_3$ as $2 \times 2$ matrices
\noindent
\[
  S_0 = \ket{0}\bra{0},\qquad
  S_1 = \ket{0}\bra{1},\qquad
  S_2 = \ket{1}\bra{0},\qquad
  S_3 = \ket{1}\bra{1}
\]
, respectively.
We define the Schmidt rank of an $n$-partite matrix $U$, acting on the Hilbert space $\mathcal{H}_1 \otimes \cdots \otimes \mathcal{H}_n \cong \mathbb{C}^{d_1} \otimes \cdots \otimes \mathbb{C}^{d_n}$, as the minimum integer $r$ such that $U$ can be expressed as
$U = \sum_{j=1}^r A_{j,1} \otimes \cdots \otimes A_{j,n}$, where each $A_{j,i}$ is a $d_i \times d_i$ matrix. We denote this $n$-partite Schmidt rank by $\sr(U)$. For any bipartition $S \cup \bar{S} = \{1,\dots,n\}$ with $S \cap \bar{S} \neq \varnothing$, let $\sr_{S \mid \bar{S}}(U)$ be the bipartite Schmidt rank. Then $\sr(U) \ge \sr_{S \mid \bar{S}}(U)$, because any $n$-partite Schmidt decomposition of $U$ yields a bipartite decomposition of $U$ across $S\mid\bar{S}$ with the same number of terms. Hence, the bipartite Schmidt rank provides a lower bound for the $n$-partite Schmidt rank.
Furthermore, we denote by $\mathbb{U}(n)$ the group of $n \times n$ unitary matrices.

We additionally denote $F_0$, $F_1$ as $2 \times 2$ matrices
\[
F_0 = \Fzero{\gamma}{\delta}
    = \begin{pmatrix}
        \cos\gamma\cos\delta & \cos\gamma\sin\delta \\
        \sin\gamma\cos\delta & \sin\gamma\sin\delta
      \end{pmatrix},
\]
\[
F_1 = \Fone{\gamma}{\delta}
    = \begin{pmatrix}
        \sin\gamma\sin\delta & -\sin\gamma\cos\delta \\
        -\cos\gamma\sin\delta & \cos\gamma\cos\delta
      \end{pmatrix}
\], where $\gamma,\delta\in\mathbb{R}$.

To systematically derive the Schmidt rank, we review a fact from \cite{landsberg2011tensors}.

\begin{lemma}\label{lem:pr}
Let $ U = \sum_{j=1}^{r} A_j \otimes R_j $ be a tripartite matrix on $ \mathcal{H}_A \otimes \mathcal{H}_B \otimes \mathcal{H}_C $, where the operators $ A_j $ on $ \mathcal{H}_A$ are linearly independent. Then the Schmidt rank $ \sr(U) $ is the minimal number of product matrices whose span contains $ \Span\{R_1,\dots,R_r\} $.
\end{lemma}

\begin{lemma}\label{lem:srlb}
Let 
\[
U = \sum_{i=0}^{m-1} X_i \otimes \left( \sum_{j=0}^{n_i-1} Y_{ij} \otimes Z_{ij} \right)
\]
be a tripartite operator on $\mathcal{H}_A \otimes \mathcal{H}_B \otimes \mathcal{H}_C$, where each $X_i$, $Y_{ij}$, $Z_{ij}$ acts on $\mathcal{H}_A$, $\mathcal{H}_B$, $\mathcal{H}_C$, respectively. If
\begin{enumerate}[label=(\alph*)]
    \item for each $i$, the sets $\{Y_{ij}\}_{0 \le j \le n_i-1}$ and $\{Z_{ij}\}_{0 \le j \le n_i-1}$ are linearly independent, and
    \item the operators $X_0,\dots,X_{m-1}$ are linearly independent,
\end{enumerate}
then there exist operators $C_0,\dots,C_{\sr(U)-1}$ on $\mathcal{H}_C$ such that
\[
\Span\{ Z_{ij} : 0 \le i < m,\; 0 \le j < n_i \} = \Span\{ C_0,\dots,C_{\sr(U)-1} \}.
\]
Consequently,
\[
\dim \Span\{ Z_{ij} : 0 \le i < m,\; 0 \le j < n_i \} \le \sr(U)
\]
and, analogously,
\[
\dim \Span\{ Y_{ij} : 0 \le i < m,\; 0 \le j < n_i \} \le \sr(U).
\]
\end{lemma}

\begin{proof}
Since $ \{X_i\} $ are linearly independent, by Lemma~\ref{lem:pr} there exist product operators $ B_k \otimes C_k $ for $ k = 0,\dots,\sr(U)-1 $ such that for each $ i $
\[
\sum_{j=0}^{n_i-1} Y_{ij} \otimes Z_{ij} = \sum_{k=0}^{\sr(U)-1} a_{ik} B_k \otimes C_k
\]
where each $a_{ik}\in\mathbb{C}$, and each $B_k$ and $C_k$ acts on $\mathcal{H}_B$ and $\mathcal{H}_C$, respectively. By linear independence of $\{Y_{ij}\}$ and $\{Z_{ij}\}$ for each i, each $ Y_{ij} $ belongs to $ \Span\{ B_k \} $ and each $ Z_{ij} $ belongs to $ \Span\{ C_k \} $ for all $ i,j $. Consequently,
\[
\dim \Span\{ Y_{ij} \} \le \sr(U),\qquad
\dim \Span\{ Z_{ij} \} \le \sr(U).
\]

Moreover, since each $ Z_{ij} \in \Span\{ C_k \} $, we have $ \Span\{ Z_{ij} \} \subseteq \Span\{ C_k \} $. For the reverse inclusion, note that
\[
U = \sum_{i=0}^{m-1} X_i \otimes \left( \sum_{k=0}^{\sr(U)-1} a_{ik} B_k \otimes C_k \right) = \sum_{k=0}^{\sr(U)-1} \left( \sum_{i=0}^{m-1} a_{ik} X_i\right) \otimes B_k \otimes C_k
\]
is a decomposition of $ U $ into $ \sr(U) $ terms. Therefore, the set
\[
\left\{ \sum_{i=0}^{m-1} a_{ik} X_i \otimes B_k : k = 0,\dots,\sr(U)-1 \right\}
\]
is linearly independent. Hence, each $ C_k $ lies in $ \Span\{ Z_{ij} \} $, giving $ \Span\{ C_k \} \subseteq \Span\{ Z_{ij} \} $. Therefore, equality holds.
\end{proof}

\Cref{lem:srlb} is central to establishing lower bounds on the Schmidt rank, as illustrated in Case~\ref{case:dim3} of \Cref{thm:U2}.

\section{Characterization of three-qubit unitary gates with Schmidt rank three}
\label{sec:sr3chara}
The goal of this section is to establish conditions under which the three-qubit unitary operator of the following form has Schmidt rank three:
\begin{equation}
U = S_0 \otimes P + S_3 \otimes Q, \label{eq:U=SP+SQ}
\end{equation}
where $Q = \ket{v_0}\bra{w_0} \otimes V_0 + \ket{v_1}\bra{w_1} \otimes V_1$, the sets $\{\ket{v_0},\ket{v_1}\}$, and $\{\ket{w_0},\ket{w_1}\}$ are orthonormal bases of $\mathbb{C}^2$, while $V_0,V_1 \in \mathbb{U}(2)$ are linearly independent and $P \in \mathbb{U}(4)$. Throughout this section, we refer to $P$ and $Q$ as the target operators corresponding to $S_0$ and $S_3$, respectively.

Since the Schmidt rank is invariant under local unitary operations, we may without loss of generality simplify these expressions \cite{KAK,Nielsen_2003,Cohen_2013}. Specifically, we can transform $U$ into $U_1$ when $\sr(P)=1$, and into $U_2$ when $\sr(P)=2$:
\begin{align}
U_1 &= S_0 \otimes I_2 \otimes I_2 + S_3 \otimes \bigl( F_0 \otimes W_0 + F_1 \otimes W_1 \bigr), \label{eq:U1_simplified} \\
U_2 &= S_0 \otimes \bigl(S_0 \otimes Z_0 + S_3 \otimes Z_1 \bigr) + S_3 \otimes \bigl( F_0 \otimes Z_2 + F_1 \otimes Z_3 \bigr), \label{eq:U2_simplified}
\end{align}
where $W_0,W_1 \in \mathbb{U}(2)$, $Z_0 = I_2$, $Z_1 = \begin{pmatrix} 1 & 0 \\ 0 & e^{i\varphi} \end{pmatrix}$ with $\varphi \notin 2\pi\mathbb{Z}$, $Z_2, Z_3 \in \mathbb{U}(2)$, and the sets $\{F_0,F_1\}$ (defined in \Cref{sec:preliminaries}), $\{W_0,W_1\}$, $\{Z_0,Z_1\}$, $\{Z_2,Z_3\}$ are each linearly independent.
We note that for $U_1$, the set $\{I_2,W_0,W_1\}$ contains the independent subset $\{W_0,W_1\}$; similarly for $U_2$, the sets $\{Z_0,Z_1\}$ and $\{Z_2,Z_3\}$ are each independent. Consequently, by Lemma~\ref{lem:srlb} we have $\sr(U_1) \geq 2$ and $\sr(U_2) \geq 2$. This fact will be used throughout without further explanation. Moreover, if either $P$ or $Q$ has Schmidt rank four, then Lemma~\ref{lem:srlb} implies $\sr(U) \geq 4$. Since our focus is on gates of Schmidt rank three, such cases are not of interest here.

We present the following lemma that will be used in the proof of \Cref{thm:U1}.
\begin{lemma}\label{lem:angle1}
For $F_0$ and $F_1$ defined in \Cref{sec:preliminaries} the following statements are equivalent:
\begin{enumerate}[label=(\roman*)]
\item\label{cond:span} The identity matrix $I_2$ belongs to the span of $\{F_0,F_1\}$.
\item\label{cond:pd_angle} $\gamma - \delta = n\pi$ for some integer $n$.
\item\label{cond:sum} $F_0 + F_1 = (-1)^n I_2$ for some integer $n$.
\end{enumerate}
\end{lemma}

\begin{proof}
We prove the equivalence in three parts.

\paragraph{\ref{cond:pd_angle}$\Leftrightarrow$\ref{cond:sum}.}
A direct computation gives
\[
F_0+F_1 = 
\begin{pmatrix}
\cos(\gamma-\delta) & -\sin(\gamma-\delta) \\
\sin(\gamma-\delta) &  \cos(\gamma-\delta)
\end{pmatrix}.
\]
This matrix equals $(-1)^n I_2$ iff $\cos(\gamma-\delta)=(-1)^n$ and $\sin(\gamma-\delta)=0$, i.e., $\gamma-\delta=n\pi$. Hence \ref{cond:pd_angle} and \ref{cond:sum} are equivalent.

\paragraph{\ref{cond:span}$\Rightarrow$\ref{cond:pd_angle}.}
Assume there exist real numbers $a$ and $b$, not both zero, such that $a F_0 + b F_1 = I_2$. Writing the four entries yields the system
\begin{align}
a\cos\gamma\cos\delta + b\sin\gamma\sin\delta &= 1, \label{eq:first}\\
a\cos\gamma\sin\delta - b\sin\gamma\cos\delta &= 0, \label{eq:second}\\
a\sin\gamma\cos\delta - b\cos\gamma\sin\delta &= 0, \label{eq:third}\\
a\sin\gamma\sin\delta + b\cos\gamma\cos\delta &= 1. \label{eq:fourth}
\end{align}
Equations \eqref{eq:second} and \eqref{eq:third} form a homogeneous linear system in $a,b$. For a non‑trivial solution its determinant must vanish:
\[
\det = \sin(\gamma-\delta)\sin(\gamma+\delta)=0.
\]
Thus either $\gamma-\delta=n\pi$ or $\gamma+\delta=m\pi$ ($n,m\in\mathbb{Z}$).  
Now we subtract \eqref{eq:fourth} from \eqref{eq:first} to obtain
\begin{equation}
(a-b)\cos(\gamma+\delta)=0, \label{eq:diff}
\end{equation}
and we add \eqref{eq:first} and \eqref{eq:fourth} to obtain
\begin{equation}
(a+b)\cos(\gamma-\delta)=2. \label{eq:sum}
\end{equation}

\begin{itemize}
\item If $\gamma-\delta=n\pi$, then $\cos(\gamma-\delta)=(-1)^n$ and from \eqref{eq:sum} we have $(a+b)(-1)^n=2$, which is solvable; moreover $\sin(\gamma-\delta)=0$ makes \eqref{eq:second} and \eqref{eq:third} automatically satisfied. So this case is possible.
\item If $\gamma+\delta=m\pi$ but $\gamma-\delta\neq n\pi$, then $\cos(\gamma+\delta)=(-1)^m$ and from \eqref{eq:diff} we obtain $(a-b)(-1)^m=0$, forcing $a=b$. Substituting $a=b$ into \eqref{eq:second} gives $a\sin(\delta-\gamma)=0$, i.e., $a\sin(\gamma-\delta)=0$. Since $a\neq0$ (otherwise $b=0$), we obtain $\sin(\gamma-\delta)=0$, contradicting $\gamma-\delta\neq n\pi$. Hence this case cannot occur.
\end{itemize}
Therefore a necessary condition for \ref{cond:span} is $\gamma-\delta=n\pi$.

\paragraph{\ref{cond:pd_angle}$\Rightarrow$\ref{cond:span}.}
Since \ref{cond:pd_angle} and \ref{cond:sum} are equivalent, condition \ref{cond:pd_angle} implies $F_0+F_1 = (-1)^n I_2$. Then taking the linear combination with coefficients $(-1)^n$ for both matrices yields
\[
(-1)^n F_0 + (-1)^n F_1 = (-1)^n (F_0+F_1) = (-1)^n\bigl((-1)^n I_2\bigr) = I_2,
\]
so $I_2$ lies in the span of $\{F_0,F_1\}$. Hence \ref{cond:pd_angle} implies \ref{cond:span}.

Thus the three conditions are equivalent. \qedhere
\end{proof}

We present the following theorem characterizing three-qubit unitary gates of the form \eqref{eq:U1_simplified} that have Schmidt rank three.

\begin{theorem} \label{thm:U1}
For the operator $U_1$ defined in \eqref{eq:U1_simplified}, we have $\sr(U_1)=3$ if and only if $\sr(U_1)\neq 2$. Moreover, $\sr(U_1)=2$ if and only if the following three conditions hold:
\begin{enumerate}[label=(\arabic*)]
    \item\label{cond:dim} $\dim\Span\{F_0,F_1,I_2\}=2$;
    \item\label{cond:comm} $W_0$ and $W_1$ commute;
    \item\label{cond:Z_diff} $W_0 - W_1$ is not a scalar multiple of the identity.
\end{enumerate}
\end{theorem}

\begin{proof}
The operator $U_1$ is a sum of three product terms,
\[
U_1 = S_0\otimes I_2\otimes I_2 \;+\; S_3\otimes F_0\otimes W_0 \;+\; S_3\otimes F_1\otimes W_1,
\]
hence $\sr(U_1)\le 3$. Furthermore, by \Cref{lem:srlb} and the independence of $W_0,W_1$, $\sr(U_1) \geq \dim\Span\{W_0,W_1,I_2\} \geq 2$. Therefore, $\sr(U_1)=3 \Longleftrightarrow \sr(U_1) \neq 2$.

We now prove the equivalence for $\sr(U_1)=2$.

\paragraph{($\Leftarrow$)}
Assume conditions~\ref{cond:dim}, \ref{cond:comm}, and \ref{cond:Z_diff}. From condition~\ref{cond:dim} and \Cref{lem:angle1} we have $F_0+F_1 = (-1)^n I_2$ for some integer $n$.

Because $W_0$ and $W_1$ commute (condition~\ref{cond:comm}) and are unitary, they are simultaneously diagonalizable: there exists a unitary $V$ such that
\[
V^*W_0V = \diag(\alpha_0,\alpha_1),\qquad 
V^*W_1V = \diag(\beta_0,\beta_1),
\]
with $|\alpha_j|=|\beta_j|=1$. Linear independence of $W_0$ and $W_1$ implies that the vectors $(\alpha_0,\beta_0)$ and $(\alpha_1,\beta_1)$ are linearly independent; hence $\Delta := \alpha_0\beta_1 - \alpha_1\beta_0 \neq 0$.  
We claim that there exist $a_0,a_1\in\mathbb{C}$ such that $I_2 = a_0W_0 + a_1W_1$ and $a_0+a_1\neq0$. Indeed, in the diagonal basis $I_2 = a_0W_0 + a_1W_1$ becomes the system
\[
\alpha_j a_0 + \beta_j a_1 = 1 \qquad (j=0,1).
\]
Therefore,
\[
a_0 + a_1 = \frac{(\alpha_0 - \beta_0) - (\alpha_1 - \beta_1)}{\Delta}.
\]
Condition~\ref{cond:Z_diff} means $\alpha_0 - \beta_0 \neq \alpha_1 - \beta_1$; therefore $a_0 + a_1 \neq 0$.

Now we have $I_2 = a_0W_0 + a_1W_1$ with $a_0+a_1\neq0$. Then
\[
\begin{aligned}
U_1 &= S_0\otimes I_2\otimes (a_0W_0 + a_1W_1) + S_3\otimes\bigl(F_0\otimes (W_0-W_1) + (-1)^n I_2\otimes W_1\bigr).
\end{aligned}
\]
We rewrite as
\[
U_1 = \bigl(S_0 + \frac{(-1)^n}{a_0+a_1}S_3\bigr)\otimes I_2\otimes (a_0W_0+a_1W_1)
    \;+\; S_3\otimes\bigl(F_0 - \frac{(-1)^n a_0}{a_0+a_1}I_2\bigr)\otimes (W_0-W_1).
\]
This expresses $U_1$ as a sum of two product terms, hence $\sr(U_1)\le2$. Furthermore, since $\sr(U_1)\neq1$, we obtain $\sr(U_1)=2$.

\paragraph{($\Rightarrow$)}
We now assume $\sr(U_1)=2$. By \Cref{lem:pr}, there exist linearly independent product operators $B_0\otimes C_0$ and $B_1\otimes C_1$ such that $P$ and $Q$ as defined in \eqref{eq:U=SP+SQ} are linear combinations of them.

By \Cref{lem:srlb} and independence of $F_0,F_1$, we know $2\leq\dim\Span\{F_0,F_1,I_2\}\leq2$. Hence $\dim\Span\{F_0,F_1,I_2\}=2$, which is condition~\ref{cond:dim}. By \Cref{lem:angle1}, this implies $F_0+F_1 = (-1)^n I_2$ for some integer $n$. Similarly, $\dim\Span\{W_0,W_1,I_2\}=2$, so we can write $I_2 = a_0 W_0 + a_1 W_1$ for some $a_0,a_1\in\mathbb{C}$.

Next we show that $a_0+a_1\neq0$. We assume for contradiction that $a_0+a_1=0$.

The set $\{B_i\otimes C_i\}_{0\le i\le 1}$ is a basis of $\Span\{B_i\otimes C_i\}_{0\le i\le 1}$. Since $I_2\otimes I_2$ lies in this span, by the Steinitz exchange lemma we may replace an operator in $\{B_i\otimes C_i\}_{0\le i\le 1}$ to obtain another basis. Thus we assume without loss of generality that
\[
\{I_2\otimes I_2,\; B_0\otimes C_0\}
\]
is a basis of $\Span\{B_i\otimes C_i\}_{0\le i\le 1}$. We can express $Q$ in this new basis:
\[
F_0\otimes W_0 + F_1\otimes W_1 = q_0 \, I_2\otimes I_2 + q_1 \, B_0\otimes C_0
\]
for some scalars $q_0,q_1\in\mathbb{C}$.

By $a_0+a_1=0$, we have
\[
F_0\otimes W_0 + F_1\otimes W_1 = q_0 a_0 I_2\otimes \bigl(W_0-W_1\bigr) + q_1 B_0\otimes C_0.
\]

Since $\{W_0,W_1\}$ are linearly independent, there exist linear functionals $\varphi_0,\varphi_1$ such that $\varphi_j(Z_k)=\delta_{jk}$. Applying $\operatorname{id}_B\otimes\varphi_0$ and $\operatorname{id}_B\otimes\varphi_1$ to the above yields
\[
F_0 = q_0 a_0 I_2 + q_1 B_0 \varphi_0(C_0),\qquad 
F_1 = -q_0 a_0 I_2 + q_1 B_0 \varphi_1(C_0).
\]
Adding these equations gives
\[
F_0 + F_1 = q_1 B_0 (\varphi_0(C_0) + \varphi_1(C_0)).
\]
By \Cref{lem:angle1}, $F_0+F_1 = (-1)^n I_2$, hence $B_0$ is a scalar multiple of $I_2$. Substituting $B_0 \propto I_2$ into the expressions for $F_0$ and $F_1$ shows that $F_0,F_1$ are both multiples of $I_2$, contradicting condition~\ref{cond:dim}. Therefore $a_0+a_1\neq0$.

Now we have $I_2 = a_0W_0 + a_1W_1$ with $a_0+a_1\neq0$. From this we deduce that $W_0$ and $W_1$ commute (condition~\ref{cond:comm}): if $a_0\neq0$ then $W_0 = (I_2 - a_1W_1)/a_0$ commutes with $W_1$; if $a_0=0$ then $I_2 = a_1W_1$ forces $W_1$ to be scalar, which trivially commutes.  

To see that $W_0-W_1$ is not a scalar multiple of the identity (condition~\ref{cond:Z_diff}), we assume for contradiction that $W_0 - W_1 = \lambda I_2$. Substituting $W_0 = W_1 + \lambda I_2$ into $I_2 = a_0W_0 + a_1W_1$ gives
\[
I_2 = a_0(W_1+\lambda I_2) + a_1W_1 = (a_0+a_1)W_1 + a_0\lambda I_2.
\]
Thus $(a_0+a_1)W_1 = (1 - a_0\lambda)I_2$. Since $a_0+a_1\neq0$, $W_1$ is a scalar multiple of $I_2$. Then $W_0 = W_1 + \lambda I_2$ is also scalar, contradicting the linear independence of $W_0$ and $W_1$. Hence condition~\ref{cond:Z_diff} holds.

Thus we have shown that conditions~\ref{cond:dim}, \ref{cond:comm}, and \ref{cond:Z_diff} are necessary and sufficient for $\sr(U_1)=2$, completing the proof.
\end{proof}

\begin{cor}
Every controlled-controlled three-qubit unitary gate has Schmidt rank two.
\end{cor}
\begin{proof}
Any controlled-controlled three qubit unitary gate has the form $S_0\otimes I_4+S_3\otimes(S_0\otimes I_2+S_3\otimes W_1)$, which is a special case of $U_1$ in \eqref{eq:U1_simplified} with $W_0=I_2$. Conditions \ref{cond:dim} and \ref{cond:comm} of \Cref{thm:U1} are satisfied. If $W_1$ is a scalar multiple of the identity, then $\sr(U_{cc})=2$; otherwise, \ref{cond:Z_diff} holds and \Cref{thm:U1} yields $\sr(U_{cc})=2$.
\end{proof}
The corollary applies to any gate that is either exactly of the controlled-controlled form or locally equivalent to one, such as the Toffoli gate and the Quantum OR gate \cite{Cohen_2013}. The statement generalizes straightforwardly: any $n$-qubit unitary controlled gate with a single target qubit has Schmidt rank two. Indeed, such a gate can be written as $U_{cc} = I_{2^n} + (\bigotimes_{i=1}^{n-1} S_3) \otimes (W - I_2)$, where $W \in \mathbb{U}(2)$, and $W$ and $I_2$ are linearly independent. 
Hence $W - I_2$ and $I_2$ are linearly independent. Together with linear independence between $\bigotimes_{i=1}^{n-1} S_3$ and $I_{2^{n-1}}$, we have $\sr(U_{cc}) \ge 2$, and therefore $\sr(U_{cc}) = 2$.

We present the following lemma that will be used in the proof of \Cref{thm:U2}.
\begin{lemma}\label{lem:angle2}
Let $F_0,F_1$ be defined as in \Cref{sec:preliminaries}. Then the matrices $F_0,F_1,S_0,S_3$ are linearly dependent if and only if one of the following conditions holds:
\begin{enumerate}[label=(\roman*)]
\item\label{cond:sr2_angle1} $\gamma-\delta = n\pi$ for some $n\in\mathbb{Z}$, and $\sin2\delta\neq0$, which yields \eqref{eq:sr2_angle1}:
\begin{equation} \label{eq:sr2_angle1}
    F_0+F_1=(-1)^nI_2.
\end{equation}
\item\label{cond:sr2_angle2} $\gamma+\delta = m\pi$ for some $m\in\mathbb{Z}$, and $\sin2\gamma\neq0$, which yields \eqref{eq:sr2_angle2}:
\begin{equation} \label{eq:sr2_angle2}
    F_0-F_1=(-1)^m(S_0-S_3).
\end{equation}
\item\label{cond:sr2_angle3} $\gamma=k\frac{\pi}{2},\delta=l\frac{\pi}{2}$ for some $k,l\in\mathbb{Z}$, such that $k$ and $l$ are both even or both odd. This case yields \eqref{eq:sr2_angle3}:
\begin{equation} \label{eq:sr2_angle3}
    \{F_0,F_1\}=\{cS_0,cS_3\} \text{ with } c=\pm1.
\end{equation} 
\end{enumerate}
\end{lemma}

\begin{proof}
We first show that linear dependence of $F_0,F_1,S_0,S_3$ is equivalent to $\sin(\gamma-\delta)\sin(\gamma+\delta)=0$. Linear dependence requires nonzero $(a,b,x,y)$ satisfying
\begin{equation}\label{eq:dependence}
a F_0 + b F_1 = x S_0 + y S_3.
\end{equation}
Expanding the matrix entries yields the system
\begin{align}
a\cos\gamma\cos\delta + b\sin\gamma\sin\delta &= x, \label{eq:angle2_syseq1}\\
a\cos\gamma\sin\delta - b\sin\gamma\cos\delta &= 0, \label{eq:angle2_syseq2}\\
a\sin\gamma\cos\delta - b\cos\gamma\sin\delta &= 0, \label{eq:angle2_syseq3}\\
a\sin\gamma\sin\delta + b\cos\gamma\cos\delta &= y. \label{eq:angle2_syseq4}
\end{align}

A nontrivial solution to \eqref{eq:dependence} exists if and only if equations \eqref{eq:angle2_syseq2} and \eqref{eq:angle2_syseq3} admit a nontrivial solution $(a,b)$. This occurs precisely when their determinant vanishes:
\begin{equation}\label{eq:det}
\sin^2\gamma\cos^2\delta - \cos^2\gamma\sin^2\delta
      = \sin(\gamma-\delta)\sin(\gamma+\delta)=0.
\end{equation}
Thus linear dependence holds if and only if \eqref{eq:det} is satisfied.

We now prove that $\gamma,\delta$ satisfy \eqref{eq:det} if and only if they fall into one of the three cases.

($\Rightarrow$) Assume \eqref{eq:det} holds.
\begin{itemize}
    \item If $\sin(\gamma-\delta)=0$, then $\gamma-\delta=n\pi$. 
    \begin{itemize}
        \item When $\sin2\delta\neq0$, we obtain Case~\ref{cond:sr2_angle1}.
        \item When $\sin2\delta=0$, we have $\delta=l\pi/2$, hence $\gamma=\delta+n\pi=(l+2n)\pi/2$. Since $l$ and $l+2n$ share the same parity, case~\ref{cond:sr2_angle3} follows.
    \end{itemize}
    \item If $\sin(\gamma+\delta)=0$, then $\gamma+\delta=m\pi$.
    \begin{itemize}
        \item When $\sin2\gamma\neq0$, we obtain Case~\ref{cond:sr2_angle2}.
        \item When $\sin2\gamma=0$, we have $\gamma=k\pi/2$, hence $\delta=m\pi-\gamma=(2m-k)\pi/2$. Since $k$ and $2m-k$ share the same parity, case~\ref{cond:sr2_angle3} follows.
    \end{itemize}
\end{itemize}

($\Leftarrow$) Conversely, we assume that $\gamma,\delta$ satisfy one of the three cases. Cases~\ref{cond:sr2_angle1} and~\ref{cond:sr2_angle2} clearly satisfy \eqref{eq:det}. For case~\ref{cond:sr2_angle3}, \eqref{eq:det} becomes $\sin(\frac{k-l}{2}\pi)\sin(\frac{k+l}{2}\pi)=0$ because $k+l$ is even.

The three cases are mutually exclusive: cases~\ref{cond:sr2_angle1} and~\ref{cond:sr2_angle2} cannot coincide, as their conjunction would imply $2\gamma=(m+n)\pi$, contradicting $\sin2\gamma\neq0$. Case~\ref{cond:sr2_angle3} requires $\sin2\gamma=\sin2\delta=0$, which contradicts the nonvanishing sine conditions in cases~\ref{cond:sr2_angle1} and~\ref{cond:sr2_angle2}.

A straightforward calculation verifies the relations \eqref{eq:sr2_angle1}, \eqref{eq:sr2_angle2}, and \eqref{eq:sr2_angle3} in their respective cases.
\end{proof}

The following corollary provides a more detailed characterization of $F_0$ and $F_1$, which will be useful in the proof of \Cref{thm:U2}.

\begin{cor}\label{cor:nonzerooffdiag}
In Cases~\ref{cond:sr2_angle1} and~\ref{cond:sr2_angle2} of Lemma~\ref{lem:angle2}:
\begin{enumerate}[label=(\alph*)]
    \item All off-diagonal entries of $F_0$ and $F_1$ are nonzero.
    \item For any $u, v \in \mathbb{C}$, if $uF_0 + vF_1$ is diagonal, then $u = v$ in Case~\ref{cond:sr2_angle1}, and $u = -v$ in Case~\ref{cond:sr2_angle2}.
\end{enumerate}
\end{cor}

\begin{proof}
In Case~\ref{cond:sr2_angle1}, $\gamma-\delta=n\pi$, $\sin2\delta\neq0$; in Case~\ref{cond:sr2_angle2}, $\gamma+\delta=m\pi$, $\sin2\gamma\neq0$. Hence all off-diagonal entries of $F_0,F_1$ are nonzero, proving (a).

For (b), let $uF_0+vF_1$ be diagonal. Using Lemma~\ref{lem:angle2}:
\begin{itemize}
    \item In Case~\ref{cond:sr2_angle1}, $F_1 = (-1)^n I_2 - F_0$, so $uF_0+vF_1 = (u-v)F_0 + (-1)^n v I_2$. Since $F_0$ has nonzero off-diagonal entries, $u-v=0$.
    \item In Case~\ref{cond:sr2_angle2}, $F_1 = (-1)^m(S_3-S_0) + F_0$, so $uF_0+vF_1 = (u+v)F_0 + (-1)^m v (S_3-S_0)$. Since $F_0$ has nonzero off-diagonal entries, $u+v=0$.
\end{itemize}
Thus $u=v$ (Case~\ref{cond:sr2_angle1}) or $u=-v$ (Case~\ref{cond:sr2_angle2}).
\end{proof}

We present the following theorem characterizing three-qubit unitary gates of the form \eqref{eq:U2_simplified} that have Schmidt rank three.

\begin{theorem}\label{thm:U2}
Let $U_2$ be given by \eqref{eq:U2_simplified}. Denote each $Z_k=\begin{pmatrix}z_{k0}&z_{k1}\\ z_{k2}&z_{k3}\end{pmatrix}$. Then the Schmidt rank of $U_2$ is determined as follows.

\newcounter{dim_counter}
\setcounter{dim_counter}{0}
\newcounter{angle_counter}
\setcounter{angle_counter}{0}

\noindent
\refstepcounter{dim_counter}
\textbf{Case \arabic{dim_counter}: $\dim\Span\{Z_0,Z_1,Z_2,Z_3\}=2$.}\label{case:dim2}
\clickableitem{pr:dim2_angle1}{cond:dim2_angle1}{Under condition \ref{cond:sr2_angle1} of Lemma~\ref{lem:angle2}, $\sr(U_2)=3$.}

\clickableitem{pr:dim2_angle2}{cond:dim2_angle2}{Under condition \ref{cond:sr2_angle2} of Lemma~\ref{lem:angle2}, $\sr(U_2)=3$.}

\clickableitem{pr:dim2_angle3}{cond:dim2_angle3}{Under condition \ref{cond:sr2_angle3} of Lemma~\ref{lem:angle2}, $\sr(U_2)=3$ if and only if
        \[
        \bigl(z_{23} - z_{20} e^{i\varphi} - z_{30} + z_{33}\bigr)^2 + 4(z_{20} - z_{23})(z_{33} - z_{30} e^{i\varphi}) = 0
        \]
        and $(z_{20}, z_{23}, z_{30}, z_{33})$ does not simultaneously satisfy
        \[
        z_{20} = z_{23},\quad z_{33} = z_{30} e^{i\varphi},\quad z_{20} = z_{30}.
        \]}

\noindent
\setcounter{angle_counter}{0}
\refstepcounter{dim_counter}
\textbf{Case \arabic{dim_counter}: $\dim\Span\{Z_0,Z_1,Z_2,Z_3\}=3$.} \label{case:dim3}

\clickableitem{pr:dim3_angle1}{cond:dim3_angle1}{Under condition \ref{cond:sr2_angle1} of Lemma~\ref{lem:angle2}, $\sr(U_2)=3$ if and only if
        \[
        z_{31}\neq z_{21},\quad z_{30}z_{21}\neq z_{31}z_{20}.
        \]}
        
\clickableitem{pr:dim3_angle2}{cond:dim3_angle2} {Under condition \ref{cond:sr2_angle2} of Lemma~\ref{lem:angle2}, $\sr(U_2)=3$ if and only if         
        \[
        z_{31} \neq -z_{21},\qquad (1+e^{i\varphi})(z_{30}z_{21}-z_{31}z_{20}) \neq 2(z_{33}z_{21}-z_{31}z_{23}).
        \]}
        
\clickableitem{pr:dim3_angle3}{cond:dim3_angle3} {Under condition \ref{cond:sr2_angle3} of Lemma~\ref{lem:angle2}, $\sr(U_2)=3$.}

\end{theorem}
\begin{proof}
We put the proof of \Cref{thm:U2} in \Cref{app:proofU2}.
\end{proof}

\begin{example} \label{examples}
We now present examples of three-qubit unitary gates with Schmidt rank three. Equations \eqref{U3}--\eqref{U5} and \eqref{U12} are constructed via \Cref{thm:U1}; notably, $U_4$ is the Peres gate \cite{PhysRevA.32.3266}. Equations \eqref{U6}--\eqref{U11} are constructed via \Cref{thm:U2}. Specifically, $U_6$, $U_7$, $U_8$ correspond to \caseref{case:dim2}{cond:dim2_angle1}, \caseref{case:dim2}{cond:dim2_angle2}, \caseref{case:dim2}{cond:dim2_angle3}, respectively, while $U_9$, $U_{10}$, $U_{11}$ correspond to \caseref{case:dim3}{cond:dim3_angle1}, \caseref{case:dim3}{cond:dim3_angle2}, \caseref{case:dim3}{cond:dim3_angle3}, respectively. The diagonal gates $U_5$ and $U_8$ can be independently verified to have Schmidt rank three via the method of \cite{Shen_2022}.

\begin{align}
    U_3 &= S_0 \otimes I_2 \otimes I_2 + S_3 \otimes (S_0 \otimes \sigma_x + S_3 \otimes \sigma_z), \label{U3} \\
    U_4 &= S_0 \otimes I_2 \otimes I_2 + S_3 \otimes (S_1 \otimes \sigma_x + S_2 \otimes I_2), \label{U4}
\end{align}
\begin{equation}\label{U5}
\begin{aligned}
U_5(\psi,\delta) = &\; S_0 \otimes I_2 \otimes I_2 \\
&+ e^{i\psi} S_3 \otimes \Biggl[ S_0 \otimes \begin{pmatrix} i e^{i\delta} + 2\sin\delta & 0 \\ 0 & -i e^{-i\delta} + 2\sin\delta \end{pmatrix} \\
&+ S_3 \otimes \begin{pmatrix} i e^{i\delta} & 0 \\ 0 & -i e^{-i\delta} \end{pmatrix} \Biggr],
\end{aligned}
\end{equation}

where $\psi,\delta \in \mathbb{R}$, and $\delta \notin \frac{\pi}{2}\mathbb{Z}$.

\begin{align}
    U_6 &= S_0 \otimes (S_0 \otimes I_2 + S_3 \otimes \sigma_z)+S_3 \otimes (\ket{+}\bra{+} \otimes I_2 + \ket{-}\bra{-} \otimes \sigma_z), \label{U6} \\
    U_7 &= S_0 \otimes (S_0 \otimes I_2 + S_3 \otimes \sigma_z)+S_3 \otimes (\ket{+}\bra{-} \otimes \sigma_z + \ket{-}\bra{+} \otimes I_2), \label{U7} \\
    U_8 &= S_0 \otimes (S_0 \otimes I_2 + S_3 \otimes \begin{bmatrix}1&0\\ 0&i\end{bmatrix})+S_3 \otimes (S_0 \otimes I_2 + S_3 \otimes \begin{bmatrix}-i&0\\ 0&-1\end{bmatrix}), \label{U8} \\
    U_9 &= S_0 \otimes (S_0 \otimes I_2 + S_3 \otimes \sigma_z)+S_3 \otimes (\ket{+}\bra{+} \otimes I_2 + \ket{-}\bra{-} \otimes \sigma_x), \label{U9} \\
    U_{10} &= S_0 \otimes (S_0 \otimes I_2 + S_3 \otimes \sigma_z)+S_3 \otimes (\ket{+}\bra{-} \otimes \sigma_x + \ket{-}\bra{+} \otimes I_2), \label{U10} \\
    U_{11} &= S_0 \otimes (S_0 \otimes I_2 + S_3 \otimes \sigma_z)+S_3 \otimes (S_0 \otimes I_2 + S_3 \otimes \sigma_x). \label{U11} \\
    U_{12} &= S_0 \otimes I_2 \otimes I_2 + S_3 \otimes (\ket{+}\bra{-} \otimes \sigma_x + \ket{-}\bra{+} \otimes \sigma_z) \label{U12}
\end{align}

\end{example}

The following theorem generalizes $U_3$ to $n$ qubits:

\begin{theorem} \label{thm:nbit}
Let \(K_1, \dots, K_{n-1}, L_2, \dots, L_{n-1}, M_2, \dots, M_{n-1} \in \mathbb{U}(2)\). Define
\[
U_{13} = S_0 \otimes \left( \bigotimes_{i=1}^{n-1} K_i \right) + S_3 \otimes \left( S_0 \otimes \left( \bigotimes_{j=2}^{n-1} L_j \right) + S_3 \otimes \left( \bigotimes_{k=2}^{n-1} M_k \right) \right),
\]
If there exists an index \(q \in \{2,\dots,n-1\}\) such that $\{K_q,L_q,M_q\}$ is a linearly independent set,
then \(U_{13}\) is unitary and its Schmidt rank is three.
\end{theorem}

\begin{proof}
The operators $S_0 \otimes K_1$, $S_3 \otimes S_0$, and $S_3 \otimes S_3$ are linearly independent. 
If additionally $\{K_q, L_q, M_q\}$ is linearly independent for some $q \in \{2,\dots,n-1\}$, then the bipartite Schmidt rank of $U_{13}$ across the partition $\{1\} \mid \{2,\dots,n-1\}$ is three. Consequently,
\[
3 = \operatorname{sr}_{\{1\} \mid \{2,\dots,n-1\}}(U_{13}) \le \operatorname{sr}(U_{13}) \le 3,
\]
where the upper bound follows from the explicit expansion of $U_{13}$. Hence $\operatorname{sr}(U_{13}) = 3$.
\end{proof}

\section{Entanglement capability of three-qubit unitary gates with Schmidt rank three}
\label{sec:entanglement}
We analyze the entanglement capabilities and $l_1$-norm of quantum coherence ($\mathcal{C}_{l_1}$) for the states generated by applying gates in \Cref{examples} to $\ket{+++}$.

\begin{table}[htbp] \label{tab:entanglement_metrics}
\centering
\begin{tabular}{lcccccc}
\hline
\textbf{Gate} & $\boldsymbol{\tau_{A(BC)}}$ & $\boldsymbol{\tau_{AB}}$ & $\boldsymbol{\tau_{AC}}$ & $\boldsymbol{\tau_{ABC}}$ & $\boldsymbol{\mathcal{C}_{l_1}}$ & \textbf{Class} \\ 
\hline
$U_3$ & 0.75 & 0.25 & 0.25 & 0.25 & 7.00 & GHZ-class \\
$U_4$ & 0 & 0 & 0 & 0 & 7.00 & Fully Separable \\
$U_5$ & 1 & $\sin^2\delta$ & $\cos^2\delta$ & 0 & 7.00 & W-class \\
$U_6\text{ to }U_7$ & 0.75 & 0.25 & 0.25 & 0.25 & 7.00 & GHZ-class \\
$U_8$ & 0.75 & 0.5 & 0.25 & 0 & 7.00 & W-class \\
$U_9\text{ to }U_{11}$ & 0.75 & 0.25 & 0.25 & 0.25 & 7.00 & GHZ-class \\
$U_{12}$ & 1 & 0 & 0 & 1 & 7.00 & Maximal GHZ \\
\hline
\end{tabular}
\caption{Tangles, coherence $\mathcal{C}_{l_1}$, and classes for $U_k\ket{+++}$ with $k=3,\dots,12$.}
\label{tab:entanglement_metrics}
\end{table}

Table~\ref{tab:entanglement_metrics} summarizes the calculated tangles~\cite{3tangle}, $l_1$-norm of coherence ($\mathcal{C}_{l_1}$), and their corresponding SLOCC entanglement classes~\cite{GHZWclass}. Crucially, despite sharing an operator Schmidt rank of three, these unitary gates generate states spanning the entire tripartite tangle range ($\tau_{ABC} \in [0, 1]$). This variation demonstrates that operator-level Schmidt rank does not uniquely dictate state-level multipartite entanglement. Notably, $U_5$ satisfies conditions~\ref{cond:dim} and~\ref{cond:comm} but violates condition~\ref{cond:Z_diff}. It confines the output state to the generalized $W$-class ($\tau_{ABC}=0$), which cannot perform deterministic teleportation. However, $U_5$ functions as a bipartite entanglement router: the parameter $\delta$ continuously modulates entanglement distribution between $\tau_{AB}$ and $\tau_{AC}$ while preserving the global boundary entanglement ($\tau_{A(BC)} = 1$). Meanwhile, the invariant $l_1$-norm of coherence ($\mathcal{C}_{l_1} = 7.00$) indicates computational-basis amplitude conservation, confirming that these gates reallocate initial superpositions without altering global coherence. The following theorem generalizes the behavior of $U_5$ to a broader class of unitaries.

\begin{theorem}\label{thm:U5general}
If Conditions \ref{cond:dim} and \ref{cond:comm} hold but Condition \ref{cond:Z_diff} fails for $U_1$ in \eqref{eq:U1_simplified}, the tripartite tangle $\tau_{ABC}$ of $U_1\ket{+++}$ vanishes.
\end{theorem}
\begin{proof}
By Condition \ref{cond:comm}, $W_0$ and $W_1$ commute and are simultaneously diagonalizable. Via a local unitary on the third qubit, $W_0 = \operatorname{diag}(\mu_0, \mu_1)$ and $W_1 = \operatorname{diag}(\lambda_0, \lambda_1)$, for some $\mu_0,\mu_1,\lambda_0,\lambda_1 \in \mathbb{C}$. The failure of Condition \ref{cond:Z_diff} implies a uniform spectral shift $\lambda_i = \mu_i + c$ for $c \in \mathbb{C}$.
By Condition \ref{cond:dim} and Lemma \ref{lem:angle1}, a local unitary on the second qubit allows us to set $F_0 = S_0$, $F_1 = S_3$, yielding:
\[U_1 = S_0 \otimes I_4 + S_3 \otimes (S_0 \otimes W_0 + S_3 \otimes W_1).\]
Direct calculation confirms that $\tau_{ABC}$ for $U_1\ket{+++}$ is zero.
\end{proof}

The following theorem generalizes $U_8$ (which falls into \caseref{case:dim2}{cond:dim2_angle3}) and shows that in this case, Schmidt rank three forces $U_2\ket{+++}$ into the W-class.
\begin{theorem}\label{thm:U8general}
In \caseref{case:dim2}{cond:dim2_angle3}, if $\operatorname{sr}(U_2)=3$, then the tripartite tangle $\tau_{ABC}$ of $U_2\ket{+++}$ vanishes.
\end{theorem}
\begin{proof}
In this case, we have
\[
U_2 = S_0 \otimes \bigl( S_0 \otimes I_2 + S_3 \otimes \begin{pmatrix} 1 & 0 \\ 0 & e^{i\varphi} \end{pmatrix} \bigr) + S_3 \otimes \bigl( S_0 \otimes (a_0 Z_0 + a_1 Z_1) + S_3 \otimes (b_0 Z_0 + b_1 Z_1) \bigr)
\]
for some $a_0,a_1,b_0,b_1 \in \mathbb{C}$. Direct calculation shows that the three-tangle of $U_2\ket{+++}$ is
\[
\tau_{ABC} = \frac{1}{16} |e^{i\varphi} - 1|^2 \cdot \bigl| (a_0 - b_1)^2 + 4a_1 b_0 \bigr|.
\]
As demonstrated at the end of the proof of \caseref{case:dim2}{cond:dim2_angle3}, the condition $\operatorname{sr}(U_2)=3$ implies $(a_0 - b_1)^2 + 4a_1 b_0 = 0$. Therefore $\tau_{ABC} = 0$.
\end{proof}

In contrast to $U_5$ and $U_8$, the tripartite tangle is maximized by generalizing $U_{12}$ to $U_{12}'$:
\begin{equation}
U_{12}' = S_0 \otimes I_4 + S_3 \otimes (\ket{+}\bra{-} \otimes V + \ket{-}\bra{+} \otimes W),
\end{equation}
where $V \in \mathbb{U}(2)$ and $W = e^{i\alpha}\ket{-}\bra{+} + e^{i\beta}\ket{+}\bra{-}$ ($\alpha, \beta \in \mathbb{R}$). Here, $U_{12}'\ket{+++}$ yields $\tau_{ABC} = 1$. Notably, $\sr(U_{12}') \neq 3$ if $V$ and $W$ become linearly dependent, further illustrating the decoupling between operator-level Schmidt rank and state-level multipartite entanglement. The resulting maximal GHZ-class state remains an optimal resource for quantum teleportation.

The results of this section reveal a nontrivial connection between the Schmidt rank conditions (see \Cref{thm:U8general}) and output state entanglement. The family $U_5(\psi,\delta)$ shows that fixed Schmidt rank gates can continuously redistribute bipartite entanglement. Characterizing the general relationship between operator-level Schmidt rank and entanglement capability thus remains an interesting open direction.

\section{Implementation of three-qubit unitary gates of Schmidt rank three}
\label{sec:sr3decomp}
We decompose unitary gates in \Cref{examples} into quantum circuits using techniques in \cite{Shende_2006,Chen_2015,bullock2003smallercircuitsarbitrarynqubit}. Let $q_0,q_1,q_2$ be the three qubits. In our diagrams, the top line corresponds to $q_0$ and the bottom line to $q_2$, and basis states are written as $\ket{q_2q_1q_0}$. We use the following notation: $CCZ_{(ab)\rightarrow c}$ denotes a Toffoli gate with controls $a,b$ and target $c$; $CZ_{a,b}$ and $CH_{a,b}$ denote controlled-$Z$ and controlled-Hadamard gates with control $a$ and target $b$, respectively; $R_z(\theta)=\diag(e^{-i\theta/2},e^{i\theta/2})$; $X=\sigma_x$; $\sqrt{X}$ is its square root; and $H = \frac{1}{\sqrt{2}}\begin{pmatrix} 1 & 1 \\ 1 & -1 \end{pmatrix}$ is the Hadamard gate. For instance, $CNOT_{2,0}=S_0\otimes I_4+S_3\otimes I_2 \otimes \sigma_x$ has its control on $q_2$ (bottom line) and target on $q_0$ (top line).

Notably, $U_3$ shown in ~\Cref{fig:U3}, $U_6$ in \eqref{U6_decomp}, and $U_{11}$ shown in \Cref{fig:U11} each use only three CNOT gates. For $U_6$ in \eqref{U6_decomp}, the decomposition consists of $CH$, $CZ$, and $CH$ gates, each requiring one CNOT, totaling three. This establishes the existence of a Schmidt-rank-three three-qubit unitary gate implementable with three CNOT gates, which is a question left open in \cite{Song_2021}.

According to \cite{shende2008cnotcosttoffoligates}, the diagonal gates $U_5$ and $U_8$ can be shown to require at least four CNOT gates; our decompositions in Figs.~\ref{fig:U8} and~\ref{fig:U5} achieve this minimum.
Peres gate $U_4$ requires at least five CNOT gates \cite{shende2008cnotcosttoffoligates}.

We also present decompositions for $U_7$ in \eqref{U7_decomp}, $U_9$ in \eqref{U9_decomp}, and $U_{10}$ in \eqref{U10_decomp}, where $U_7$ and $U_{10}$ use the Peres and Toffoli gates as building blocks.

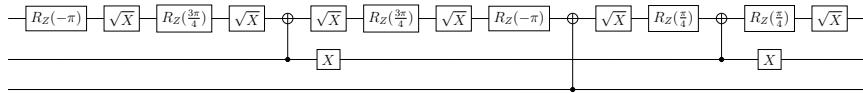
\begin{figure}[htbp]
\centering
\large
\resizebox{0.95\textwidth}{!}{%
\begin{quantikz}
& \gate{R_Z(-\pi)} & \gate{\sqrt{X}} & \gate{R_Z(\frac{3\pi}{4})} & \gate{\sqrt{X}} & \targ{} & \gate{\sqrt{X}} & \gate{R_Z(\frac{3\pi}{4})} & \gate{\sqrt{X}} & \gate{R_Z(-\pi)} & \targ{} & \gate{\sqrt{X}} & \gate{R_Z(\frac{\pi}{4})} & \targ{} & \gate{R_Z(\frac{\pi}{4})} & \gate{\sqrt{X}} & \qw\\
& \qw & \qw & \qw & \qw & \ctrl{-1} & \gate{X} & \qw & \qw & \qw & \qw & \qw & \qw & \ctrl{-1} & \gate{X} & \qw & \qw\\
& \qw & \qw & \qw & \qw & \qw & \qw & \qw & \qw & \qw & \ctrl{-2} & \qw & \qw & \qw & \qw & \qw & \qw
\end{quantikz}
}
\caption{Quantum circuit for $U_3$ defined in \eqref{U3}.}
\label{fig:U3}
\end{figure}

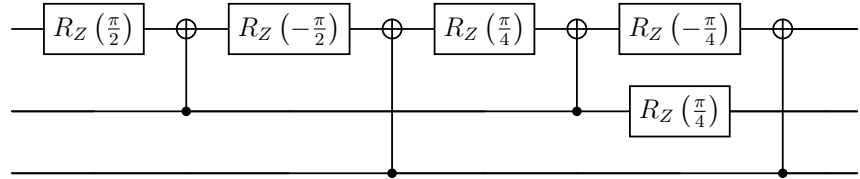
\begin{figure}[htbp]
\centering
\large
\resizebox{0.95\textwidth}{!}{%
\begin{quantikz}
& \gate{R_Z\left(\frac{\pi}{2}\right)} & \targ{} & \gate{R_Z\left(-\frac{\pi}{2}\right)} & \targ{} & \gate{R_Z\left(\frac{\pi}{4}\right)} & \targ{} & \gate{R_Z\left(-\frac{\pi}{4}\right)} & \targ{} & \qw & \qw \\
& \qw & \ctrl{-1} & \qw & \qw & \qw & \ctrl{-1} & \gate{R_Z\left(\frac{\pi}{4}\right)} & \qw & \qw & \qw \\
& \qw & \qw & \qw & \ctrl{-2} & \qw & \qw & \qw & \ctrl{-2} & \qw & \qw
\end{quantikz}
}
\caption{Quantum circuit for $U_8$ defined in \eqref{U8}.}
\label{fig:U8}
\end{figure}

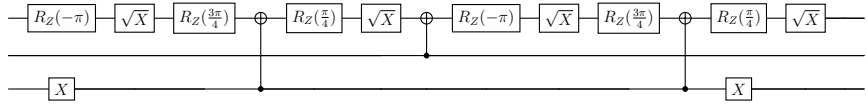
\begin{figure}[htbp]
\centering
\large
\resizebox{0.95\textwidth}{!}{%
\begin{quantikz}
& \gate{R_Z(-\pi)} & \gate{\sqrt{X}} & \gate{R_Z(\frac{3\pi}{4})} & \targ{} & \gate{R_Z(\frac{\pi}{4})} & \gate{\sqrt{X}} & \targ{} & \gate{R_Z(-\pi)} & \gate{\sqrt{X}} & \gate{R_Z(\frac{3\pi}{4})} & \targ{} & \gate{R_Z(\frac{\pi}{4})} & \gate{\sqrt{X}} & \qw & \qw\\
& \qw & \qw & \qw & \qw & \qw & \qw & \ctrl{-1} & \qw & \qw & \qw & \qw & \qw & \qw & \qw & \qw\\
& \gate{X} & \qw & \qw & \ctrl{-2} & \qw & \qw & \qw & \qw & \qw & \qw & \ctrl{-2} & \gate{X} & \qw & \qw & \qw
\end{quantikz}
}
\caption{Quantum circuit for $U_{11}$ defined in \eqref{U11}.}
\label{fig:U11}
\end{figure}

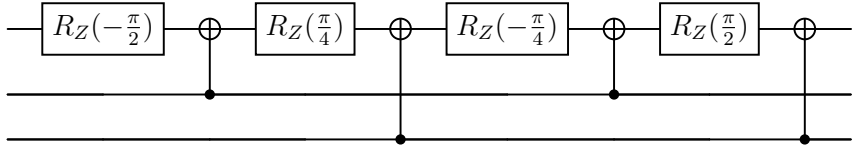
\begin{figure}[!htb]
\centering
\large
\resizebox{0.95\textwidth}{!}{%
\begin{quantikz}
& \gate{R_Z(-\frac{\pi}{2})} & \targ{} & \gate{R_Z(\frac{\pi}{4})} & \targ{} & \gate{R_Z(-\frac{\pi}{4})} & \targ{} & \gate{R_Z(\frac{\pi}{2})} & \targ{} & \qw \\
& \qw & \ctrl{-1} & \qw & \qw & \qw & \ctrl{-1} & \qw & \qw & \qw \\
& \qw & \qw & \qw & \ctrl{-2} & \qw & \qw & \qw & \ctrl{-2} & \qw
\end{quantikz}
}
\caption{Quantum circuit for $U_5$ defined in \eqref{U5}.}
\label{fig:U5}
\end{figure}

\begin{align}
    U_6 &= CH_{2,1} \, CZ_{1,0} \, CH_{2,1} \label{U6_decomp} \\
    U_7 &= (\sigma_x \otimes I_4) \, CCZ_{(2,1)\rightarrow0} \, (\sigma_x \otimes I_4) \, (I_2 \otimes H \otimes H) \, U_4 \, (I_2 \otimes H \otimes H) \label{U7_decomp} \\
    U_9 &= CH_{2,0} CH_{2,1} CZ_{1,0} CH_{2,0} CH_{2,1} \label{U9_decomp} \\
    U_{10} &= (\sigma_x \otimes I_4) \, CCZ_{(2,1)\rightarrow0} \, (\sigma_x \otimes I_4) \, (I_2 \otimes H \otimes I_2) \, U_4 \, (I_2 \otimes H \otimes I_2) \label{U10_decomp}
\end{align}

\section{Conclusion}
\label{sec:conclusion}In this paper, we characterized three-qubit unitary operators of Schmidt rank three that are in controlled form, establishing necessary and sufficient conditions for such gates to have Schmidt rank three. We provided explicit examples, and showed that some of these gates admit optimal decompositions using only three CNOT gates, achieving the known lower bound. Several interesting directions for future research include, but are not limited to: determining which three-qubit unitaries of Schmidt rank three can be implemented with exactly three CNOT gates; extending this characterization to 
n-qubit unitary gates of Schmidt rank three; and characterizing Schmidt-rank-three unitaries that are not in controlled form.

\begin{appendices}
\section{Proof of \Cref{thm:U2}}\label[appendix]{app:proofU2}
\begin{proof}
To simplify, we first put $U_2$ into a convenient form.
In Case \ref{case:dim2}, $Z_2,Z_3\in\Span\{Z_0,Z_1\}$; write $Z_2=a_0Z_0+a_1Z_1$ and $Z_3=b_0Z_0+b_1Z_1$.
In Case \ref{case:dim3}, we may assume without loss of generality that $Z_3=t_0Z_0+t_1Z_1+t_2Z_2$. 
Indeed, if instead $Z_2$ is a combination of $\{Z_0,Z_1,Z_3\}$, applying $\begin{pmatrix}0&-1\\1&0\end{pmatrix}$ on the left and $\begin{pmatrix}0&1\\-1&0\end{pmatrix}$ on the right of the middle system swaps $F_0$ and $F_1$, thereby exchanging $Z_2$ and $Z_3$. 
If $Z_0$ or $Z_1$ lies in the span of the other three, applying $\begin{pmatrix}\cos\gamma&\sin\gamma\\\sin\gamma&-\cos\gamma\end{pmatrix}$ on the left and $\begin{pmatrix}\cos\delta&\sin\delta\\\sin\delta&-\cos\delta\end{pmatrix}$ on the right of the middle system interchanges $\{F_0,F_1\}$ with $\{S_0,S_3\}$, allowing us to relabel so that $Z_3$ is the dependent one.
Thus we may always write $Z_3=t_0Z_0+t_1Z_1+t_2Z_2$ with $\{Z_0,Z_1,Z_2\}$ independent.

\hypertarget{pr:dim2_angle1}{}
\medskip
\noindent\textbf{\caseref{case:dim2}{cond:dim2_angle1} } 
Using \eqref{eq:sr2_angle1}, $F_1 = (-1)^n I_2 - F_0$. $P,Q$ from \eqref{eq:U=SP+SQ} become
\begin{align}
P &= S_0 \otimes Z_0 + S_3 \otimes Z_1,\\
Q &=  F_0 \otimes (a_0Z_0+a_1Z_1) + [(-1)^n I_2 - F_0] \otimes (b_0Z_0+b_1Z_1)\label{eq:Q_decomp_1i}.
\end{align}
We first show $\sr(U_1)\le 3$ by considering three cases for $a_i,b_i$:

\begin{enumerate}[label=(\alph*)]
\item $b_0 \neq -b_1$:
\begin{equation}
\label{eq:U_case1}
\begin{aligned}
U_2 &= \bigl(\frac{S_0}{b_0+b_1} + (-1)^n S_3\bigr) \otimes I_2 \otimes (b_0Z_0+b_1Z_1) \\
  &\quad + S_0 \otimes \bigl(S_0 - \frac{b_0}{b_0+b_1} I_2\bigr) \otimes (Z_0 - Z_1) \\
  &\quad + S_3 \otimes F_0 \otimes \bigl((a_0-b_0)Z_0 + (a_1-b_1)Z_1\bigr).
\end{aligned}
\end{equation}

\item $b_0 = -b_1$ but $a_0 \neq -a_1$:
\begin{equation}
\label{eq:U_case2}
\begin{aligned}
U_2 &= \bigl(\frac{S_0}{a_0+a_1} + (-1)^n S_3\bigr) \otimes I_2 \otimes (a_0Z_0+a_1Z_1) \\
  &\quad + S_0 \otimes \bigl(S_3 - \frac{a_1}{a_0+a_1} I_2\bigr) \otimes (Z_1 - Z_0) \\
  &\quad - S_3 \otimes F_1 \otimes \bigl((a_0-b_0)Z_0 + (a_1+b_0)Z_1\bigr).
\end{aligned}
\end{equation}

\item $b_0 = -b_1$ and $a_0 = -a_1$:
Then $Z_2 = a_0(Z_0 - Z_1)$ and $Z_3 = b_0(Z_0 - Z_1)$ are proportional, contradicting linear independence of $Z_2,Z_3$. Thus this case cannot occur.
\end{enumerate}

Every possible case gives three product terms, hence $\sr(U_1) \leq 3$.

Now $\sr(U_1)=1$ is impossible since $\dim\Span\{Z_0, Z_1, Z_2, Z_3\}=2$. We assume for contradiction that $\sr(U_1)=2$. By \Cref{lem:pr}, there exist linearly independent product operators $B_0\otimes C_0$ and $B_1\otimes C_1$ such that $P$ and $Q$ are linear combinations of them:

\begin{equation}\label{eq:PQBC_1i}
P = \sum_{i=0}^{1} p_i B_i \otimes C_i, \quad
Q = \sum_{i=0}^{1} q_i B_i \otimes C_i.
\end{equation}

Because $\sr(P)=2$, both $p_0$ and $p_1$ are nonzero. Moreover, by \Cref{lem:srlb}, $\Span\{C_0,C_1\} = \Span\{Z_0,Z_1\}$; therefore, $C_0,C_1$ are linearly independent.

We write each $B_i$ as $B_i = \begin{pmatrix} b_{i0} & b_{i1} \\ b_{i2} & b_{i3} \end{pmatrix}$. Since the off-diagonal blocks of $P$ vanish by \eqref{eq:U2_simplified}, from \eqref{eq:PQBC_1i}:
\[
\sum_{i=0}^{1} p_i b_{i1} C_i = 0,\qquad \sum_{i=0}^{1} p_i b_{i2} C_i = 0.
\]
Independence of $C_0,C_1$ together with $p_i\neq0$ forces $b_{i1}=b_{i2}=0$ for each $i$; hence each $B_i$ is diagonal: $B_i=\diag(b_{i0},b_{i3})$.

Let $F_0 = \begin{pmatrix} f_{00} & f_{01} \\ f_{02} & f_{03} \end{pmatrix}$. Comparing the top-right block of $Q$ in \eqref{eq:Q_decomp_1i} and \eqref{eq:PQBC_1i} together with the diagonality of $B_i$, we obtain
$f_{01}\bigl((a_0-b_0)Z_0 + (a_1-b_1)Z_1\bigr)=0$. Hence, independence of $Z_0,Z_1$ gives $f_{01}(a_0-b_0)=0$ and $f_{01}(a_1-b_1)=0$. Since $f_{01}\neq0$ by \Cref{cor:nonzerooffdiag}, we obtain $a_0=b_0,\; a_1=b_1$. Then $Z_2 = Z_3$, contradicting the linear independence of $Z_2,Z_3$. Therefore $\sr(U_1)=2$ is impossible.

Since $\sr(U_1)\le 3$ and neither $1$ nor $2$ can occur, we must have $\sr(U_1)=3$.

\hypertarget{pr:dim2_angle2}{}
\medskip
\noindent\textbf{\caseref{case:dim2}{cond:dim2_angle2} } 
Using \eqref{eq:sr2_angle2}, $F_0 - (-1)^m(S_0 - S_3) = F_1$. $P,Q$ from \eqref{eq:U=SP+SQ} become
\begin{align}
P &= S_0 \otimes Z_0 + S_3 \otimes Z_1,\\
Q &= F_0 \otimes (a_0Z_0+a_1Z_1) + [F_0-(-1)^m(S_0-S_3)] \otimes (b_0Z_0+b_1Z_1).
\end{align}

We first show $\sr(U_1)\le 3$ by considering three cases for $a_i,b_i$:

\begin{enumerate}[label=(\alph*)]
\item $b_0 \neq b_1$:
\begin{equation}
\label{eq:U_case2_1}
\begin{aligned}
U_2 &= \bigl(\frac{S_0}{b_1-b_0} + (-1)^m S_3\bigr) \otimes (S_3-S_0) \otimes (b_0Z_0+b_1Z_1) \\
    &\quad + S_3 \otimes F_0 \otimes \bigl((a_0+b_0)Z_0 + (a_1+b_1)Z_1\bigr) \\
    &\quad + S_0 \otimes \bigl(S_0 - \frac{b_0}{b_1-b_0}(S_3-S_0)\bigr) \otimes (Z_0+Z_1).
\end{aligned}
\end{equation}

\item $b_0 = b_1$ but $a_0 \neq a_1$:
\begin{equation}
\label{eq:U_case2_2}
\begin{aligned}
U_2 &= \bigl(\frac{S_0}{a_1-a_0} + (-1)^{m+1}S_3\bigr) \otimes (S_3-S_0) \otimes (a_0Z_0+a_1Z_1) \\
    &\quad + S_3 \otimes F_1 \otimes \bigl((a_0+b_0)Z_0 + (a_1+b_1)Z_1\bigr) \\
    &\quad + S_0 \otimes \bigl(S_0 - \frac{a_0}{a_1-a_0}(S_3-S_0)\bigr) \otimes (Z_0+Z_1).
\end{aligned}
\end{equation}

\item $b_0 = b_1$ and $a_0 = a_1$:
Then $Z_2$ and $Z_3$ are both proportional to $Z_0+Z_1$, contradicting linear independence of $Z_2,Z_3$. Thus this case cannot occur.
\end{enumerate}

Every possible case gives three product terms, hence $\sr(U_1) \leq 3$.

Since $\dim\Span\{Z_0,Z_1,Z_2,Z_3\}=2$, $\sr(U_1)=1$ is impossible by \Cref{lem:srlb}. Assuming $\sr(U_1)=2$ and arguing as in \caseref{case:dim2}{cond:dim2_angle1}, we obtain $f_{01}\neq0$ and
\[
f_{01}\bigl((a_0+b_0)Z_0 + (a_1+b_1)Z_1\bigr)=0,
\]
which forces $a_0=-b_0$, $a_1=-b_1$. Thus $Z_2=-Z_3$, contradicting the independence of $Z_2,Z_3$. Therefore $\sr(U_1)\neq2$.

Since $\sr(U_1)\le 3$ and neither $1$ nor $2$ can occur, we must have $\sr(U_1)=3$.

\hypertarget{pr:dim2_angle3}{}
\medskip
\noindent\textbf{\caseref{case:dim2}{cond:dim2_angle3} } In this case we have \eqref{eq:sr2_angle3}. Without loss of generality, assume that $F_0 = S_0$ and $F_1 = S_3$. $P,Q$ from \eqref{eq:U=SP+SQ} become
\begin{equation}\label{eq:PQ_decomp_1iii}
\begin{aligned} 
P &= S_0 \otimes Z_0 + S_3 \otimes Z_1,\\
Q &= S_0 \otimes (a_0Z_0+a_1Z_1) + S_3 \otimes (b_0Z_0+b_1Z_1).
\end{aligned}
\end{equation}

We first show $\sr(U_2)\le 3$ by considering three cases for $a_i,b_i$:

\begin{enumerate}[label=(\alph*)]
\item $a_0 \neq 0$:
\begin{equation}
\label{eq:U_case3_1}
\begin{aligned}
U_2 = &\left(\frac{S_0}{a_0}+S_3\right)\otimes S_0\otimes (a_0 Z_0 + a_1 Z_1) \\
    &+ S_0\otimes \left(S_3-\frac{a_1}{a_0}S_0\right)\otimes Z_1 \\
    &+ S_3\otimes S_3\otimes (b_0 Z_0 + b_1 Z_1).
\end{aligned}
\end{equation}

\item $a_0 = 0$ but $b_1 \neq 0$:
\begin{equation}
\label{eq:U_case3_2}
\begin{aligned}
U_2 = &S_0\otimes (S_0-\tfrac{b_0}{b_1}S_3)\otimes Z_0 + S_3\otimes S_0\otimes (a_0 Z_0 + a_1 Z_1) \\
    &+ (\frac{S_0}{b_1}+S_3)\otimes S_3 \otimes (b_0 Z_0 + b_1 Z_1).
\end{aligned}
\end{equation}

\item $a_0 = 0$ and $b_1 = 0$:
\begin{equation}
\label{eq:U_case3_3}
\begin{aligned}
U_2 = &\bigl(S_0\otimes S_0 + a_1 S_3\otimes S_0\bigr)\otimes (Z_0+Z_1) \\
    &+ S_0\otimes (S_3-S_0)\otimes Z_1 \\
    &+ S_3\otimes (b_0 S_3 - a_1 S_0)\otimes Z_0.
\end{aligned}
\end{equation}
\end{enumerate}

Every possible case gives three product terms, hence $\sr(U_2) \leq 3$.

To exclude the case when $\sr(U_2)=2$, we now establish the following equivalence:

\begin{equation}\label{eq:dim2_case3_claim}
\sr(U_2)=2 \;\Longleftrightarrow\; 
\text{$\begin{pmatrix} a_0 & b_0 \\ a_1 & b_1 \end{pmatrix}$ is diagonalizable}.
\end{equation}

We first Assume $\sr(U_2)=2$. By \Cref{lem:pr}, there exist linearly independent product operators $B_0\otimes C_0$ and $B_1\otimes C_1$ such that
\begin{equation}
\label{eq:PQBC_1iii}
P = \sum_{i=0}^{1} p_i\, B_i\otimes C_i,\quad
Q = \sum_{i=0}^{1} q_i\, B_i\otimes C_i.
\end{equation}

Because $\sr(P)=2$, both $p_0$ and $p_1$ are nonzero. Moreover, by \Cref{lem:srlb}, $\Span\{C_0, C_1\} = \Span\{Z_0, Z_1\}$; hence $C_0,C_1$ are linearly independent. Hence, $\{C_0,C_1\}$\allowbreak, $\{Z_0,Z_1\}$, $\{Z_2,Z_3\}$ are the bases of the same space. Recall that $Z_2 = \sum_{i=0}^{1} a_i Z_i$ and $Z_3 = \sum_{i=0}^{1} b_i Z_i$.

Write each $B_i$ as $B_i = \begin{pmatrix} b_{i0} & b_{i1} \\ b_{i2} & b_{i3} \end{pmatrix}$. Since the off-diagonal blocks of $P$ vanish, from \eqref{eq:PQBC_1iii}:
\[
\sum_{i=0}^{1} p_i b_{i1} C_i = 0,\qquad \sum_{i=0}^{1} p_i b_{i2} C_i = 0.
\]
Independence of $C_0,C_1$ together with $p_i\neq0$ forces $b_{i1}=b_{i2}=0$ for each $i$; hence each $B_i$ is diagonal: $B_i=\diag(b_{i0},b_{i3})$.

From \eqref{eq:PQ_decomp_1iii} and \eqref{eq:PQBC_1iii} we obtain
\begin{equation}\label{eq:Z0Z1Z2Z3}
\begin{split}
Z_0 &= \sum_{i=0}^{1} p_i b_{i0} C_i, \quad
Z_1 = \sum_{i=0}^{1} p_i b_{i3} C_i, \\
Z_2 &= \sum_{i=0}^{1} q_i b_{i0} C_i, \quad
Z_3 = \sum_{i=0}^{1} q_i b_{i3} C_i.
\end{split}
\end{equation}

Since $\{Z_2, Z_3\}$ and $\{Z_0, Z_1\}$ are both bases of the same two-dimensional space, consider the change-of-basis matrix representing the identity map with domain basis $\{Z_2,Z_3\}$ and codomain basis $\{Z_0,Z_1\}$:
\begin{equation}\label{eq:Z2Z3,Z0Z1}
[I_2]_{Z_2, Z_3}^{Z_0, Z_1} = 
\begin{pmatrix}
a_0 & b_0 \\
a_1 & b_1
\end{pmatrix}.
\end{equation}

Let $P' = \diag(p_0,p_1)$, $Q' = \diag(q_0,q_1)$, and $B = \begin{pmatrix} b_{00} & b_{03} \\ b_{10} & b_{13} \end{pmatrix}$.
\[
[I_2]_{Z_0,Z_1}^{C_0,C_1} = P'B,\qquad [I_2]_{Z_2,Z_3}^{C_0,C_1} = Q'B.
\]
Since $p_i,q_i\neq0$, $P'$ and $Q'$ are invertible, hence $B$ is invertible. Therefore,
\[
[I_2]_{Z_2,Z_3}^{Z_0,Z_1} = [I_2]_{C_0,C_1}^{Z_0,Z_1}\; [I_2]_{Z_2,Z_3}^{C_0,C_1} = B^{-1} P'^{-1} Q' B,
\]
which is diagonalizable as $P'^{-1}Q'$ is diagonal.

For the converse direction of \eqref{eq:dim2_case3_claim}, we assume that \eqref{eq:Z2Z3,Z0Z1} is diagonalizable. Then there exists an invertible $B$ and nonzero $\lambda_0,\lambda_1 \in \mathbb{C}$ such that $[I_2]_{Z_2,Z_3}^{Z_0,Z_1}$ = $B^{-1} \diag(\lambda_0,\lambda_1) B$.
We now construct $p_i,q_i\in \mathbb{C}$ and  $B_i,C_i\in \mathbb{M}_2$ for $i=0,1$ satisfying \eqref{eq:PQBC_1iii}.
Choose nonzero $p_0,p_1$ and set $q_i = \lambda_i p_i$ for $i=0,1$. Define $C_0,C_1$ by
\[
C_0 = \sum_{i=0}^{1} (B^{-1}P'^{-1})_{i0} Z_i,\quad 
C_1 = \sum_{i=0}^{1} (B^{-1}P'^{-1})_{i1} Z_i,
\]
with $P' = \diag(p_0,p_1)$. Furthermore, we write $B = \begin{pmatrix} b_{00} & b_{03} \\ b_{10} & b_{13} \end{pmatrix}$ and define $B_i = \diag(b_{i0}, b_{i3})$ for $i=0,1$.

Let $Q' = \diag(q_0,q_1)$. Then one verifies that $[I_2]_{Z_0,Z_1}^{C_0,C_1} = P'B$ and $[I_2]_{Z_2,Z_3}^{C_0,C_1} = Q'B$.
Hence \eqref{eq:Z0Z1Z2Z3} holds; consequently \eqref{eq:PQBC_1iii} holds, yielding
\[
U_2 = S_0\otimes P + S_3\otimes Q = \sum_{i=0}^{1} (p_i S_0 + q_i S_3)\otimes B_i\otimes C_i,
\]
a sum of two product terms. Thus $\sr(U_2)\le 2$. Since $\sr(U_2)\neq1$, we obtain $\sr(U_2)=2$, which proves \eqref{eq:dim2_case3_claim}. Consequently, $\sr(U_2)=3$ if and only if \eqref{eq:Z2Z3,Z0Z1} is not diagonalizable.

We now exploit the specific conditions $Z_0=I_2$ and $Z_1=\diag(1,e^{i\varphi})$ with $\varphi\notin2\pi\mathbb{Z}$ in \eqref{eq:U2_simplified} to characterize when \eqref{eq:Z2Z3,Z0Z1} fails to be diagonalizable. Since $Z_2,Z_3\in\Span\{I_2,Z_1\}$, they must be diagonal; we can write $Z_2=\diag(z_{20},z_{23})$, and $Z_3=\diag(z_{30},z_{33})$, with $|z_{20}|=|z_{23}|=|z_{30}|=|z_{33}|=1$ forced by unitarity of $Z_2,Z_3$. Solving $Z_2=a_0I_2+a_1Z_1$ and $Z_3=b_0I_2+b_1Z_1$ yields
\[
a_0=\frac{z_{23}-z_{20}e^{i\varphi}}{1-e^{i\varphi}},\quad 
a_1=\frac{z_{20}-z_{23}}{1-e^{i\varphi}},\quad 
b_0=\frac{z_{33}-z_{30}e^{i\varphi}}{1-e^{i\varphi}},\quad 
b_1=\frac{z_{30}-z_{33}}{1-e^{i\varphi}}.
\]

A $2\times2$ matrix is not diagonalizable if and only if it has a repeated eigenvalue but is not a scalar matrix, i.e., the discriminant $\Delta=(a_0-b_1)^2+4a_1b_0$ of its characteristic polynomial vanishes while the matrix is not a scalar multiple of the identity. Substituting the expressions above yields $\Delta=0$ precisely when
\begin{equation}
\label{eq:non-diag-condition}
\bigl(z_{23}-z_{20}e^{i\varphi}-z_{30}+z_{33}\bigr)^2+4(z_{20}-z_{23})(z_{33}-z_{30}e^{i\varphi})=0.
\end{equation}
Moreover, \eqref{eq:Z2Z3,Z0Z1} is a scalar multiple of the identity if and only if $a_1=0$, $b_0=0$, and $a_0=b_1$, which is equivalent to
\[
z_{20}=z_{23},\quad z_{33}=z_{30}e^{i\varphi},\quad z_{20}=z_{30}.
\]
Therefore, \eqref{eq:Z2Z3,Z0Z1} is not diagonalizable if and only if condition \eqref{eq:non-diag-condition} holds and $z_{20},z_{23},\allowbreak z_{30},z_{33}$ do not satisfy all three equalities above simultaneously. In this situation we have $\sr(U_2)=3$.

\hypertarget{pr:dim3_angle1}{}
\noindent\textbf{\caseref{case:dim3}{cond:dim3_angle1} } 
First, using condition \ref{cond:sr2_angle1} in \Cref{lem:angle2}, we have $F_1 = (-1)^n I_2 - F_0$. $P,Q$ from \eqref{eq:U=SP+SQ} become
\begin{equation}\label{eq:PQ_decomp_2i}
\begin{aligned}
P &= S_0 \otimes Z_0 + S_3 \otimes Z_1,\\
Q &= F_0 \otimes Z_2 + F_1 \otimes Z_3 = F_0 \otimes Z_2 + [(-1)^n I_{2} - F_0] \otimes \sum_{i=0}^{2} t_i Z_i.
\end{aligned}
\end{equation}
By \Cref{lem:srlb} we know $\sr(U_2)\geq\dim\Span\{Z_0,Z_1,Z_2,Z_3\}=3$.

We write $Z_3 = t_0 Z_0 + t_1 Z_1 + t_2 Z_2$. We first prove the following equivalence:
\begin{equation}\label{cond:t0t1t2}
    \sr(U_2)=3 \Longleftrightarrow t_0 + t_1 \neq 0\text{ and }t_2 \neq 1.
\end{equation}

\paragraph{($\Leftarrow$)}
We assume $t_0 + t_1 \neq 0$ and $t_2 \neq 1$. We show $\sr(U_2)=3$ by proving $\sr(U_2)\leq3$ in two cases.

\begin{enumerate}[label=(\alph*)]
\item $t_1\neq0$:
\begin{equation}
\begin{aligned}
U_2&=\bigl(\frac{S_0}{t_0+t_1}+\frac{(-1)^{n+1}S_3}{t_2-1}\bigr) \otimes I_2 \otimes (t_0Z_0+t_1Z_1)\\
&\quad+\frac{S_0}{t_1} \otimes (t_1S_0-t_0S_3) \otimes (Z_0-\frac{t_0Z_0+t_1Z_1}{t_0+t_1})\\
&\quad+S_3 \otimes \bigl(\frac{(-1)^nI_2}{t_2-1}+F_1\bigr) \otimes \bigl(t_0Z_0+t_1Z_1+(t_2-1)Z_2\bigr)
\end{aligned}
\end{equation}

\item $t_1 = 0$ (note that $t_0 \neq 0$, otherwise $\sr(Q) = 1$):
\begin{equation}
\begin{aligned}
U_2&=(S_0-\frac{t_0(-1)^n}{t_2-1}S_3) \otimes I_2 \otimes Z_0 \\
&\quad+ S_0 \otimes S_3 \otimes (Z_1-Z_0) \\
&\quad+ S_3 \otimes ((-1)^n+(t_2-1)F_1)\otimes(Z_2+\frac{t_0Z_0}{t_2-1})
\end{aligned}
\end{equation}
\end{enumerate}
Every possible case gives three product terms, hence $\sr(U_2) \leq 3$.

\paragraph{($\Rightarrow$)}
Assume $\sr(U_2) = 3$. Then there exist three linearly independent product operators $\{B_i \otimes C_i\}_{0 \leq i \leq 2}$ such that
\[
P = \sum_{i=0}^2 p_i\, B_i\otimes C_i,\qquad Q = \sum_{i=0}^2 q_i\, B_i\otimes C_i.
\]
By \Cref{lem:srlb}, $\Span\{C_i\}_{0\leq i\leq2} = \Span\{Z_i\}_{0\leq i\leq2}$. Since $\{Z_i\}_{0\leq i\leq2}$ is independent, both sets are bases. Thus we have $C_i = \sum_{j=0}^2 c_{ij} Z_j$, with each $c_{ij} \in \mathbb{C}$. We define the change-of-basis matrix:
\begin{equation}\label{eq:C}
    C:= [c_{ij}]_{0\leq i,j \leq 2}.
\end{equation}

Substituting $C_i = \sum_{j=0}^2 c_{ij} Z_j$ into $P$ and $Q$ gives
\[
P = \sum_{i,j=0}^{2} p_i c_{ij} B_i\otimes Z_j,\qquad 
Q = \sum_{i,j=0}^{2} q_i c_{ij} B_i\otimes Z_j.
\]
Equating the above expansions with the original definitions of $P$ and $Q$ in \eqref{eq:PQ_decomp_2i}, and using the linear independence of $\{Z_j\}$, we obtain:
\begin{equation}
\begin{aligned}
S_0 &= \sum_{i=0}^2 p_i c_{i0} B_i, &
S_3 &= \sum_{i=0}^2 p_i c_{i1} B_i, &
0 &= \sum_{i=0}^2 p_i c_{i2} B_i,\\
t_0 F_1 &= \sum_{i=0}^2 q_i c_{i0} B_i, &
t_1 F_1 &= \sum_{i=0}^2 q_i c_{i1} B_i, &
F_0 + t_2 F_1 &= \sum_{i=0}^2 q_i c_{i2} B_i.
\end{aligned}
\end{equation}
From the first three equations, we have $(p_0B_0,\; p_1B_1,\; p_2B_2)C = (S_0,\; S_3,\; 0).$ Multiplying on the right by $C^{-1}$ gives, for $i=0,1,2$,
\begin{equation}\label{eq:dim3_1i_first3}
p_i B_i = (C^{-1})_{0i} S_0 + (C^{-1})_{1i} S_3. 
\end{equation}
By the last three equations, $(q_0B_0,\; q_1B_1,\; q_2B_2) = (t_0F_1,\; t_1F_1,\; F_0 + t_2F_1)C^{-1}$.
Hence, we have for $i=0,1,2$,
\begin{equation}\label{eq:dim3_1i_last3}
q_i B_i = (C^{-1})_{2i} F_0 + \bigl( t_0 (C^{-1})_{0i} + t_1 (C^{-1})_{1i} + t_2 (C^{-1})_{2i} \bigr) F_1.
\end{equation}
We assume for contradiction that $t_0+t_1=0$. For every $j\in\{0,1,2\}$, $p_j$ and $q_j$ are not both zero (otherwise $\sr(U_2)\le 2$). This allows us to consider the following three cases based on the zero/nonzero pattern of $(p_j,q_j)$. Under $t_0+t_1=0$, we will prove that \eqref{eq:0j=1j} holds for all $j=0,1,2$, contradicting the invertibility of $C$.
\begin{equation} \label{eq:0j=1j}
(C^{-1})_{0j} = (C^{-1})_{1j}.
\end{equation}

\begin{enumerate} [label=(\alph*)]
\item \(p_j \neq 0,\; q_j \neq 0\): \label{casea}
By \eqref{eq:dim3_1i_first3}, $B_j$ is diagonal. Therefore, $q_jB_j$ is diagonal. By \eqref{eq:dim3_1i_last3} and \Cref{cor:nonzerooffdiag}, we have
\begin{equation} \label{eq:sameCoeffqrBr}
(C^{-1})_{2j} = t_0 (C^{-1})_{0j} + t_1 (C^{-1})_{1j} + t_2 (C^{-1})_{2j}. \end{equation}
Consequently, by \eqref{eq:dim3_1i_last3},\eqref{eq:sameCoeffqrBr} and \eqref{eq:sr2_angle1}, we have
\begin{equation}\label{eq:qjBj}
q_j B_j = (C^{-1})_{2j} (F_0 + F_1) = (C^{-1})_{2j} (-1)^n I_2.
\end{equation}
Combining \eqref{eq:dim3_1i_first3} and \eqref{eq:qjBj} by eliminating $B_j$ we know
\[
\frac{q_j}{p_j} \bigl( (C^{-1})_{0j} S_0 + (C^{-1})_{1j} S_3 \bigr) = (C^{-1})_{2j} (-1)^n I_2.
\]
Linear independence of $S_0,S_3$ yields \eqref{eq:0j=1j}.

Moreover, $(C^{-1})_{2j} \neq 0$; otherwise $q_j B_j = 0$ by \eqref{eq:qjBj} would force $B_j=0$, contradicting independence of $\{B_i \otimes C_i\}_{0\leq i \leq 2}$. Similarly, \eqref{eq:dim3_1i_first3} and \eqref{eq:0j=1j} forces $(C^{-1})_{0j}$ and $(C^{-1})_{1j}$ to be nonzero.

From \eqref{eq:sameCoeffqrBr} and \eqref{eq:0j=1j}, we have
\begin{equation} \label{eq:2j0j}
(C^{-1})_{2j} (1 - t_2) = (C^{-1})_{0j} (t_0 + t_1).
\end{equation}

\item \(p_j \neq 0,\; q_j = 0\): \label{caseb}
By \eqref{eq:dim3_1i_last3} and independence of $F_0,F_1$, we have
$t_0 (C^{-1})_{0j} + t_1 (C^{-1})_{1j} = 0.$ Since $t_0+t_1=0$, we obtain \eqref{eq:0j=1j}.

\item \(p_j = 0,\; q_j \neq 0\): \label{casec}
By \eqref{eq:dim3_1i_first3} and independence of $S_0,S_3$, we obtain $(C^{-1})_{0j} = (C^{-1})_{1j} = 0$. Hence, \eqref{eq:0j=1j} holds.

\end{enumerate}

Hence for each column $j$ of $C^{-1}$, we have $(C^{-1})_{0j} = (C^{-1})_{1j}$. Thus the first two rows of $C^{-1}$ coincide, contradicting the invertibility of $C$. Hence $t_0 + t_1 \neq 0$.

Now we show $t_2 \neq 1$. Since neither $P$ nor $Q$ is a product operator, at least two of $\{p_i\}_{0 \leq i \leq 2}$ are nonzero and at least two of the $\{q_i\}_{0 \leq i \leq 2}$ are nonzero. Hence, by the pigeonhole principle, there exists an index $r$ such that $p_r \neq 0$ and $q_r \neq 0$, falling into Case \ref{casea}.
Since $(C^{-1})_{2r}$,$(C^{-1})_{0r}$, and $t_0+t_1$ are all nonzero, \eqref{eq:2j0j} gives $t_2 \neq 1$. This completes the proof of \eqref{cond:t0t1t2}.

We now exploit the specific conditions $Z_0=I_2$ and $Z_1=\diag(1,e^{i\varphi})$ with $\varphi\notin2\pi\mathbb{Z}$ in \eqref{eq:U2_simplified} to characterize when $t_0+t_1\neq0$ and $t_2\neq1$.

First, $Z_2$ is not diagonal. Indeed, if $Z_2$ were diagonal, then $Z_3=t_0Z_0+t_1Z_1+t_2Z_2$ would also be diagonal, implying $\{Z_2,Z_3\}\subseteq\Span\{Z_0,Z_1\}$, which contradicts $\dim\Span\{Z_0,Z_1,Z_2,Z_3\}=3$. Hence $Z_2$ is non-diagonal; in particular, $z_{21}\neq0$, where we write each $Z_k=\begin{pmatrix}z_{k0}&z_{k1}\\ z_{k2}&z_{k3}\end{pmatrix}$.

From $Z_3=t_0Z_0+t_1Z_1+t_2Z_2$ we obtain the entrywise relations
\begin{equation} \label{eq:entry_2i}
z_{31}=t_2z_{21},\qquad z_{32}=t_2z_{22},\qquad z_{30}=t_0+t_1+t_2z_{20}
\end{equation}
Since $z_{21}\neq0$, the first equation in \eqref{eq:entry_2i} implies
\[
t_2\neq1\;\Longleftrightarrow\; z_{31}\neq z_{21}.
\]
Moreover, $t_0+t_1=z_{30}-t_2z_{20}$ from the third equation in \eqref{eq:entry_2i}; hence
\[
t_0+t_1\neq0\;\Longleftrightarrow\; z_{30}\neq t_2z_{20}\;\Longleftrightarrow\; z_{30}z_{21}\neq t_2z_{21}z_{20}=z_{31}z_{20},
\]

Therefore, $t_2\neq1$ and $t_0+t_1\neq0$ are equivalent to the two conditions
\begin{equation}\label{eq:entry_cond_2i}
z_{31}\neq z_{21},\qquad z_{30}z_{21}\neq z_{31}z_{20}.
\end{equation}
Consequently, $\sr(U_2)=3$ if and only if \eqref{eq:entry_cond_2i} holds.

\hypertarget{pr:dim3_angle2}{}
\noindent\textbf{\caseref{case:dim3}{cond:dim3_angle2} } 
By \Cref{lem:srlb} we have $\sr(U_2)\geq\dim\Span\{Z_0,Z_1,Z_2,Z_3\}=3$. We first prove that $\sr(U_2)=3$ if and only if $t_2\neq-1$ and $t_0-t_1\neq0$.
\paragraph{($\Leftarrow$)}
Assume $t_0 - t_1 \neq 0$ and $t_2 \neq -1$. We will show that $\sr(U_2) = 3$. It remains to show $\sr(U_2)\leq3$. Using condition \ref{cond:sr2_angle2} in \Cref{lem:angle2}, we have $F_0 = (-1)^m (S_0-S_3) + F_1$. We now decompose $U_2$ according to two cases for $t_0$.

\begin{enumerate}[label=(\alph*)]
\item $t_0\neq0$:
\begin{equation}
\begin{aligned}
U_2&=\bigl(\frac{S_0}{t_1-t_0}+\frac{(-1)^m}{t_2+1}S_3\bigr) \otimes (S_3-S_0) \otimes (t_0Z_0+t_1Z_1)\\
&\quad+t_0S_0 \otimes \bigl(\frac{S_0}{t_0}-\frac{S_3-S_0}{t_1-t_0}\bigr) \otimes (Z_0+Z_1)\\
&\quad+S_3 \otimes (F_1+\frac{(-1)^{m+1}}{t_2+1}(S_3-S_0)) \otimes \bigl(t_0Z_0+t_1Z_1+(t_2+1)Z_2\bigr).
\end{aligned}
\end{equation}

\item $t_0 = 0$ (note that $t_1 \neq 0$, otherwise $\sr(Q) = 1$):
\begin{equation}
\begin{aligned}
U_2&=\bigl(\frac{S_0}{t_1}+\frac{(-1)^m}{t_2+1}S_3\bigr) \otimes (S_3-S_0) \otimes t_1Z_1\\
&\quad+S_0 \otimes S_0 \otimes (Z_0+Z_1)\\
&\quad+S_3 \otimes (F_1+\frac{(-1)^{m+1}}{t_2+1}(S_3-S_0)) \otimes (t_1Z_1+(t_2+1)Z_2).
\end{aligned}
\end{equation}
\end{enumerate}
Every possible case gives three product terms, hence $\sr(U_2) \leq 3$.

\paragraph{($\Rightarrow$)}
We assume for contradiction that $t_0 - t_1 = 0$ or $t_2 = -1$. By an argument similar to \caseref{case:dim3}{cond:dim3_angle1}, the first row of $C^{-1}$ (defined in \eqref{eq:C}) equals $-1$ times the second row, contradicting the invertibility of $C^{-1}$. Hence $t_0 - t_1 \neq 0$ and $t_2 \neq -1$.

Combining both directions, $\sr(U_2)=3$ if and only if $t_0-t_1\neq0$ and $t_2\neq-1$. We will use a similar argument as in \caseref{case:dim3}{cond:dim3_angle1} to derive the equivalent conditions for $t_2\neq-1$ and $t_0-t_1\neq0$. Recall that $Z_0=I_2$ and $Z_1=\diag(1,e^{i\varphi})$ with $\varphi\notin2\pi\mathbb{Z}$ in \eqref{eq:U2_simplified}. Since $Z_2$ is not diagonal (otherwise $\dim\Span\{Z_0,Z_1,Z_2,Z_3\}<3$), we have $z_{21}\neq0$. From $Z_3=t_0Z_0+t_1Z_1+t_2Z_2$ we obtain the entrywise relations
\[
z_{31}=t_2z_{21},\qquad z_{30}=t_0+t_1+t_2z_{20},\qquad z_{33}=t_0+t_1e^{i\varphi}+t_2z_{23}.
\]

By the first relation, $t_2\neq-1$ is equivalent to $z_{31}\neq -z_{21}$. Moreover, we have
\[
t_0-t_1 = \frac{(1+e^{i\varphi})(z_{30}z_{21}-z_{31}z_{20})-2(z_{33}z_{21}-z_{31}z_{23})}{z_{21}(e^{i\varphi}-1)}.
\]
Therefore, $t_0-t_1\neq0$ if and only if
\[
(1+e^{i\varphi})(z_{30}z_{21}-z_{31}z_{20}) \neq 2(z_{33}z_{21}-z_{31}z_{23}).
\]

Thus, $t_2\neq-1$ and $t_0-t_1\neq0$ are equivalent to
\begin{equation}\label{eq:entry_cond_2ii}
z_{31} \neq -z_{21},\qquad (1+e^{i\varphi})(z_{30}z_{21}-z_{31}z_{20}) \neq 2(z_{33}z_{21}-z_{31}z_{23}).
\end{equation}
As a result, we conclude that $\sr(U_2)=3$ if and only if \eqref{eq:entry_cond_2ii} holds.

\hypertarget{pr:dim3_angle3}{}
\noindent\textbf{\caseref{case:dim3}{cond:dim3_angle3} } 
By \Cref{lem:srlb} we know $\sr(U_2)\geq\dim\Span\{Z_0,Z_1,Z_2,Z_3\}=3$. It remains to show $\sr(U_2)\leq3$. Using condition \ref{cond:sr2_angle3} in \Cref{lem:angle2}, we may assume without loss of generality that $F_0 = S_0$, $F_1=S_3$. We now decompose $U_2$ according to three cases for the coefficients $t_i$.

\begin{enumerate}[label=(\alph*)]
\item $t_2\neq0$:
\begin{equation}
\begin{aligned}
U_2 =&\; S_3 \otimes (S_0 + t_2 S_3) \otimes \left(Z_2 + \frac{t_0}{t_2} Z_0\right) \\
   &+ \left(S_0 - \frac{t_0}{t_2} S_3\right) \otimes S_0 \otimes Z_0 \\
   &+ \left(S_0 + t_1 S_3\right) \otimes S_3 \otimes Z_1.
\end{aligned}
\end{equation}

\item $t_2 = 0, t_1\neq0$:
\begin{equation}
\begin{aligned}
U_2 =&\; S_0 \otimes \left(S_0 - \frac{t_0}{t_1} S_3\right) \otimes Z_0 \\
   &+ S_3 \otimes S_0 \otimes Z_2 \\
   &+ \left(\frac{1}{t_1} S_0 + S_3\right) \otimes S_3 \otimes (t_0 Z_0 + t_1 Z_1).
\end{aligned}
\end{equation}

\item $t_2 = 0, t_1=0$:
\begin{equation}
\begin{aligned}
U_2 =&\; (S_0 + t_0 S_3) \otimes (S_0 + S_3) \otimes Z_0 \\
   &+ S_0 \otimes S_3 \otimes (Z_1 - Z_0) \\
   &+ S_3 \otimes S_0 \otimes (Z_2 - t_0 Z_0).
\end{aligned}
\end{equation}
\end{enumerate}

Every possible case gives three product terms, hence $\sr(U_2) \leq 3$. Hence, $\sr(U_2)=3$.

\end{proof}
\end{appendices}

\bibliographystyle{ieeetr}
\bibliography{ref}

\end{document}